\documentclass[a4paper,onecolumn,11pt,accepted=2024-12-17]{quantumarticle}
\pdfoutput=1
\usepackage{makecell}
\usepackage{array}
\newcolumntype{C}[1]{>{\centering\arraybackslash}m{#1}}
\usepackage{booktabs,adjustbox}
\usepackage{float}
\usepackage[utf8]{inputenc}  
\usepackage[T1]{fontenc}     %Output what you want e.g., é, ł, a, ü
\usepackage[table]{xcolor}
\usepackage{color,soul}
\usepackage[british]{babel}  %Do hyphenation according to british english

\usepackage[colorlinks=true, citecolor=blue, urlcolor=blue]{hyperref}  
\usepackage{graphicx} % Package to insert exteral figures
\usepackage[babel]{microtype}  %Improves text justification
\usepackage{amsmath,amssymb,amsthm,bm,amsfonts,mathrsfs,bbm} %Usefull math packages

\usepackage{xspace}  %Useful to add space in macros
\usepackage{pgf,tikz}
\usepackage{xcolor}
\usepackage{multirow}
\usepackage{array}
\usepackage{bigstrut}
\usepackage{braket}
\usepackage{color}
\usepackage[numbers]{natbib}
\usepackage{multirow}
\usepackage{mathtools}
\usepackage{float}
\usepackage{blkarray}
\usepackage[caption = false]{subfig}
\usepackage{xcolor,colortbl}
\usepackage{color}
\usepackage{rotating}
\usepackage{tikz}
\usetikzlibrary{quantikz}
\usepackage{tabularx}
\usepackage{subfig}
\newcommand{\Tr}{\operatorname{Tr}}

\newcommand{\be}{\begin{equation}}
\newcommand{\ee}{\end{equation}}
\newcommand{\ba}{\begin{eqnarray}}
\newcommand{\ea}{\end{eqnarray}}
\newcommand{\ketbra}[2]{|#1\rangle \langle #2|}
\newcommand{\tr}{\operatorname{Tr}}

\newtheorem{theorem}{Theorem}

\newtheorem{definition}{Definition}
\newtheorem{proposition}{Proposition}

\newtheorem{lemma}{Lemma}

\def\>{\rangle}
\def\<{\langle}

\usepackage{centernot}
\usepackage{subfig}

\begin{document}

%\title{Bipartite composition of regular polygons: entanglement classes and nonlocal correlations}

\title{Bipartite polygon models: classes of entanglement and their nonlocal behaviour}

\author{Mayalakshmi Kolangatt}
\affiliation{School of Physics, IISER Thiruvananthapuram, Vithura, Kerala 695551, India.}

\author{Thigazholi Muruganandan}
\affiliation{School of Physics, IISER Thiruvananthapuram, Vithura, Kerala 695551, India.}

\author{Sahil Gopalkrishna Naik}
\affiliation{Department of Physics of Complex Systems, S.N. Bose National Center for Basic Sciences, Block JD, Sector III, Salt Lake, Kolkata 700106, India.}

\author{Tamal Guha}
\affiliation{Department of Computer Science, The University of Hong Kong, Pokfulam road 999077, Hong Kong.}

\author{Manik Banik}
%\email{manik11ju@gmail.com}
\affiliation{Department of Physics of Complex Systems, S.N. Bose National Center for Basic Sciences, Block JD, Sector III, Salt Lake, Kolkata 700106, India.}

\author{Sutapa Saha}
\affiliation{Department of Astrophysics and High Energy Physics, S.N. Bose National Center for Basic Sciences, Block JD, Sector III, Salt Lake, Kolkata 700106, India.}
\affiliation{Harish-Chandra Research Institute, HBNI, Chhatnag Road, Jhunsi, Allahabad 211 019, India.}

\begin{abstract}
Hardy's argument constitutes an elegantly logical test for identifying nonlocal features of multipartite correlations. In this paper, we investigate Hardy's nonlocal behavior within a broad class of operational theories, including the qubit state space as a specific case. Specifically, we begin by examining a wider range of operational models with state space descriptions in the form of regular polygons. First, we present a systematic method to characterize the possible forms of entangled states within bipartite compositions of these models. Then, through explicit examples, we identify the classes of entangled states that exhibit Hardy-type nonlocality. Remarkably, our findings highlight a closer analogy between odd polygon models and the qubit state space in terms of their bipartite Hardy nonlocal behavior compared to even-sided polygons. Furthermore, we demonstrate that the emergence of mixed-state Hardy nonlocality in any operational model is determined by a specific symmetry inherent in its dynamic description. Finally, our results uncover an unexplored class of almost-quantum correlations that can be associated with an explicit operational model.

%Hardy's argument constitutes an elegantly logical test for identifying {\it nonlocal} \textcolor{blue}{features} of multipartite correlations. In this paper, we explore Hardy's nonlocal behavior within a class of operational theories, notably encompassing qubit state space as a special case. More specifically, we start by considering a broader range of operational models with state space descriptions as regular polygons. Initially, we provide a systematic method to characterize the possible forms of entangled states in their bipartite compositions. Subsequently, supported by explicit examples, we identify the classes of entangled states capable of exhibiting Hardy-type nonlocality. Notably, our findings reveal a stronger analogy between odd polygon models and qubit state spaces regarding their bipartite Hardy nonlocal behavior, compared to even-sided polygons. Furthermore, we establish that the possibility of mixed-state Hardy nonlocality in any operational model is governed by a specific symmetry associated with their dynamic description. Lastly, our results point out an unexplored class of almost-quantum correlations that can be associated with an explicit operational model.
\end{abstract}

%\pacs{03.65.Ta,03.65.Ud, 03.67.Dd}
%\keywords{}

% 03.65.Ta	Foundations of quantum mechanics;
% 03.67.Dd	Quantum cryptography and communication security
% 03.67.Hk	Quantum communication
% 03.65.Ud	Entanglement and quantum nonlocality

\maketitle
%\onecolumngrid
\section{Introduction}
Entanglement, the exotic state space feature of composite quantum systems, establishes one of the most striking departures of quantum theory from the classical worldview \cite{Horodecki2009}. While marginal states obtained from maximally entangled states are always completely mixed, quantum theory also allows non-maximally entangled states to have partially mixed marginals. In the operational paradigm of local operations and classical communications (LOCC), latter states are less useful than the former one. However, there exist other operational paradigms where non-maximal states outperform the maximal ones. The study of nonlocality as initiated by the seminal Bell's theorem provides such an instance \cite{Bell64} (see also \cite{Bell66,Mermin93,Brunner14a}). While the celebrated Clauser-Horne-Shimony-Holt (CHSH) inequality is maximally violated by two-qubit maximally entangled states \cite{Clauser70,Cirel'son80}, there exist Bell inequalities, known by the name of tilted CHSH inequalities, that are maximally violated by non-maximally entangled pure two-qubit states \cite{Acin12,Coladangelo17}. In fact, there is a variant of the nonlocality test, proposed by L. Hardy \cite{Hardy93}, that establishes the nonlocal behaviour of any two-qubit pure non-maximally entangled state, whereas maximally entangled state fails the test \cite{Goldstein94, Kar97}. Subsequently, the superiority of non-maximally entangled pure states over the maximal ones has been depicted in Bayesian games \cite{Banik19a} and in reverse zero-error channel coding \cite{Alimuddin2023}.

A natural question is whether the existence of such non-maximally entangled states is truly a quantum feature or not. To address this question we study a general class of theories allowed within the framework of generalized probability theory (GPT) \cite{Janotta11,Janotta14,Brunner14,Banik19b,Bhattacharya20,Saha20(1)}. One can start with a simple toy theory with elementary state space described by a square and the composition of two such elementary systems allows an entangled state that is even more nonlocal than quantum theory \cite{Popescu94}. However, this toy model does not allow any joint state that can be thought of as an analogue to a non-maximally entangled pure state. This might lead to an intuition that non-maximal entanglement is a truly quantum signature. In the present work, first, we show that this intuition is not true and as a consequence, we come up with several important observations steamed out from it.

To explore the possibilities of non-maximally entangled pure states in beyond quantum theories, we begin with a class of GPT systems whose normalized state spaces are described by regular polygons \cite{Janotta11}. We then consider the bipartite compositions of these models and make a systemic study on the structure of allowed entangled states. Note that nonlocal properties of the correlations obtained from a particular entangled state (with maximally mixed marginals) of these systems have been studied in \cite{Janotta11}. However, we show that there exist other classes of pure entangled states that are not connected with the earlier ones via local reversible transformations. In particular, for bipartite pentagon system, we characterize the entangled states into two inequivalent classes, whereas for the bipartite hexagon system, six different entanglement classes are possible. We then study the nonlocal strength of correlations obtained from these different classes of entangled states.
It is already reported in \cite{Janotta11} that bipartite compositions of even-gon theories admit maximal CHSH violation beyond the Tsirelson's bound, which readily assures their unphysicality. However, this argument does not extend to the bipartite composition of odd-gon models as the CHSH violation in those models is upper bounded by Tsirelson's value. In contrast, our results, which deal with the Hardy-type nonlocality \cite{Hardy93}, reveal that the bipartite compositions of these operational models, with odd numbers of vertices, are able to produce beyond-quantum local input-output statistics. Moreover, we show that, unlike the two-qubit mixed entangled states, a class of mixed entangled states in such discrete operational models exhibit Hardy-type nonlocality. In addition, by putting restrictions on the maximal numbers of distinguishable pure preparations in any operational model, we figure out the kinematical topology that forbids such mixed entangled preparations of that theory to exhibit Hardy-type nonlocality.  

On the other hand, there exists a set of post-quantum correlations, namely the \textit{almost quantum correlations}, which satisfies all possible (known till now) bipartite physical principles \cite{Navascues15}. The operational models, able to mimic this particular set of correlations, are required to violate the no-restriction hypothesis \cite{Sainz18}, which simply demands that all possible mathematically consistent effects should be taken as valid observable\textcolor{blue}{s} for any GPT. However, it is not known whether a nontrivial (i.e., strictly non-quantum) proper subset of almost quantum correlations can be reproduced by an operational model, without losing the no-restriction hypothesis. In the present work, we identify a class of measure non-zero set of correlations obtained from odd-gon models, specifically from the bipartite pentagon models without invoking any restriction on them, which are strictly almost quantum correlations. Note that our result is stronger than that of \cite{Janotta11}, where Janotta \textit{et. al.} has shown that all the local correlations obtained from the maximally entangled state of odd-gon theories must lie in the set of almost quantum correlations. This is because the latter one is not established as strictly (i.e., other than quantum) almost quantum correlations.

The remaining sections of the paper are organized
as follows: in Section \ref{two} we discuss the framework and recall some relevant results required for our purpose. In Section \ref{three} we classify the entangled states, and in Section \ref{four} study their Hardy's nonlocality behaviour. As a consequence, in Section \ref{five}, we demonstrate a one-parameter correlation that strictly resides in the almost quantum realm but doesn't belong to quantum. In Section \ref{six}, the study regards the region of intersection between the correlation obtained from quantum theory and the odd-gon theory. In Section \ref{seven} we show that the concepts of entanglement and nonlocality in the polygon models are different, and finally, we put our concluding remarks in Section \ref{eight}.   

\section{Preliminaries}\label{two}
In this section, we briefly recall some basic concepts that will be required to present our main technical results. In the following, we start by describing the framework of GPT. 
\subsection{Generalized probability theories}
This mathematical framework is broad enough to encapsulate all possible probabilistic theories that use the notion of states to yield the outcome probabilities of measurements. Although the origin of this framework dates back to the nineteen sixties \cite{Ludwig,Mielni68}, it has drawn renewed interest recently. For an elaborate survey of this framework within the language of quantum information theory, we refer to the works \cite{Hardy01,Barrett07,Chiribella10,Barnum11}. In a GPT, a system $S$ is specified by identifying the three-tuple $\left(\Omega_S,\mathcal{E}_S,\mathcal{T}_S\right)$ of the state space, effect space, and the set of transformations.  

{\it State space:} $\Omega_S$ denotes the set of normalized states of the system, where a state is a mathematical object yielding outcome probabilities for all the measurements that can possibly be carried out on this system. Generally, $\Omega_S$ is considered to be a compact convex set embedded in some real vector space $V$. The extreme points of this set are called pure states. 

{\it Effect space:} An effect $e$ is a map $e:\Omega_S\to[0,1]$, where $e(\omega)$ denotes probability of obtaining the outcome corresponding to effect $e$ when a measurement is performed on the state $\omega\in\Omega_S$. The effects are considered to be linear functionals, $e[p\omega_1+(1-p)\omega_2]=pe(\omega_1)+(1-p)e(\omega_2),~\forall~\omega_1,\omega_2\in\Omega_S,~\&~\forall~p\in[0,1]$. There is a special effect, called unit effect, which is defined by $u(\omega)=1,~\forall~\omega\in\Omega_S$. The set of all proper effects is denoted as $\mathcal{E}_S\equiv\{e~|~0\le e(\omega)\le1,~\forall~\omega\in\Omega_S\}$. It is the convex hull of zero effect, unit effect, and the extremal effects and embedded in the vector space $V^\star$ dual to $V$. A measurement $\mathcal{M}$ is a collection of effects adding to unit effect, {\it i.e.} $\mathcal{M}\equiv\{e_i\in\mathcal{E}_S~|~\sum_ie_i=u\}$.

{\it State and effect cones:} Mathematically it is sometimes convenient to work with the notion of unnormalized states and effects. The set of unnormalized states forms a cone $V_+\subset V$, where $r\omega\in V_+$ for $r\ge0$ and $\omega\in\Omega_S$. The set of unnormalized effects forms a dual cone $V_+^\star\subset V^\star$, where $V_+^\star\equiv\left\{e~|~e(\omega)\ge0,~\forall~\omega\in V^\star_+\right\}$. In this cone picture, the extreme effects can be further classified as ray extremal and non-ray extremal effects. A ray $\vec{r}\in V^\star_+$ is called an extreme ray if there do not exist rays $\vec{r}_1,\vec{r}_2\in V^\star_+$ and scalar $\mu$, with $\vec{r}_1\neq\lambda\vec{r}_2$ for $\lambda>0$ and $0<\mu<1$, such that $\vec{r}=\mu\vec{r}_1+(1-\mu)\vec{r}_2$. The formulation generally assumes the `no-restriction hypothesis' which specifies the set of possible measurements to be the dual of the set of states \cite{Chiribella11}. 

{\it Transformation:} The set $\mathcal{T}_S$ denotes the collection of transformations that map states to states, {\it i.e.} $T(V_+) \subset V_+,~\forall~T\in\mathcal{T}_S$. They are also assumed to be linear to preserve statistical mixtures. They cannot increase the total probability but are allowed to decrease it. 

{\it Composite system:} Given two systems $A$ and $B$, the composite state space $\Omega_{AB}$ is embedded in the tensor product space $V_A\otimes V_B$. Although the choice of $\Omega_{AB}$ is not unique, the no signaling principle and tomographic locality postulate \cite{Hardy01} bound the choices within two extremes -- the minimal tensor product space ($\otimes_{\min}$) and maximal tensor product space ($\otimes_{\max}$) \cite{Namioka69}. More formally,
\begin{align*}
\Omega^{\min}_{AB}&\equiv\left\{\omega_{AB}=\sum p_i\omega_A^i\otimes\omega_B^i,~s.t.\right.\\
&\left.\omega_A^i\in\Omega_A,~\omega_B^i\in\Omega_B,~\&~p_i\ge0,~\sum p_i=1\right\},\\
\Omega_{AB}^{\max}&\equiv\left\{\omega_{AB},~s.t.~u_A\otimes u_B(\omega_{AB})=1,\right.\\
&\left.~e_A\otimes e_B(\omega_{AB})\ge0,~\forall~e_A\in\mathcal{E}_A~\&~e_B\in\mathcal{E}_B\right\}.
\end{align*}
In this work, we are interested in a particular class of GPT systems known as polygon models, which we recall in the following subsection. 

\subsection{Polygon models}
This class of models was first proposed in \cite{Janotta11} to understand the limited nonlocal behavior of a theory from its local state space structure. Although these are just hypothetical toy systems, in the recent past, several interesting studies have been reported on them, which subsequently bring several nontrivial foundational insights towards the structure of Hilbert space quantum theory \cite{Muller12,Massar14,Safi15,Saha20,Bhattacharya20,Saha20(1),Patra22}. Such a system we will denote as $\mathcal{P}_{ly}(n)\equiv\left(\Omega_n,\mathcal{E}_n,\mathcal{T}_n\right)$, where $n$ takes values from positive integers. 

{\it State spaces}: For every elementary system, the state space $\Omega_{n}$ is a regular polygon with $n$ sides. The state space can also be thought of as a convex hull of $n$ pure states $\{\omega_i\}_{i=1}^{n}\subset\mathbb{R}^3$, that is given by
\begin{align}\label{state}
\omega_i=\begin{pmatrix}
r_n \cos \frac{2 \pi i}{n}\\
r_n \sin \frac{2 \pi i}{n}\\
1\end{pmatrix},~~ \mbox{with}~~r_n=\sqrt{\sec(\pi/n)}.
\end{align}
For $n=1,2,3$ the state spaces are simplices, whereas interesting scenarios arise for $n\ge4$. For any such cases, the state space allows some mixed states to have non-unique decompositions in terms of extreme points. For the circular state space, the form of the pure states read as, $\omega_{\theta}=(\cos\theta,\sin\theta,1)^{\mathrm{T}}$ with $\theta$ varying continuously from $0$ to $2\pi$. 

{\it Effect spaces:} The effects are constructed in such a way that the outcome probabilities, defined as $e_i(\omega_j):=e_i^{\mathrm{T}}\omega_j$, come out to be a valid one, {\it i.e.}, $0\leq e_i(\omega_j)\leq 1,~\forall\omega_j\in\Omega_{n},e_i\in\mathcal{E}_n$. Among all the extreme points of the convex effect space $\mathcal{E}_n$, the null effect $\Theta$ and the unit effect $u$ are designated uniquely for all the models, and are given by,
\begin{align}
\Theta:=(0,0,0)^{\mathrm{T}},~~u:=(0,0,1)^{\mathrm{T}}. 
\end{align}
The other extremal effects $\{e_i\}_{i=1}^{n}$ and their complementary ones $\{\bar{e}_i:=u-e_i\}_{i=1}^{n}$ are given by:
\begin{subequations}\label{effect}
\begin{align}
e_i&:=\frac{1}{2}\begin{pmatrix}
r_n \cos \frac{(2i-1) \pi}{n}\\
r_n \sin\frac{(2i-1) \pi}{n}\\
1\end{pmatrix}~~\mbox{(for even n)},\\
e_i&:=\frac{1}{1+r_n^2}\begin{pmatrix}
r_n \cos \frac{2i\pi}{n}\\
r_n \sin\frac{2i\pi}{n}\\
1\end{pmatrix} \mbox{(for odd n)}. 
\end{align}    
\end{subequations}
For the even-gons, the complementary effects turn out to be extremal effects that are also ray extremal effects. 
However, for odd-gons, although $e_i$'s are ray extremal effects, the $\bar{e}_i$'s are not ray extremal. For circular state space, the extreme effects read as $e_i:=1/2(\cos\theta_i,\sin\theta_i,1)^{\mathrm{T}}$, with $\theta_i\in[0,2\pi]$.  

With the complete characterization of the state and effect spaces, a natural question emerges regarding the discrimination of two or more different preparations by performing the measurements allowed in that particular theory. A set of states $\Omega_{\text{PD}}:=\{\omega_1,\omega_2,\cdots,\omega_n\}$ is said to be perfectly distinguishable if there exists a measurement $\mathcal{M}:=\{e_i|\sum_{i=1}^n e_i=u\}$ such that $e_i(\omega_j)=\delta_{ij}$. The maximum cardinality of the set $\Omega_{\text{PD}}$ \textcolor{blue}{is} generally referred to as the \textit{Operational Dimension} (OD) of that particular GPT. For a single qubit as well as for all the GPT models with regular polygonal state-space description, the OD equals exactly $2$.

{\it Transformations:} For any $n$-gon, the set of reversible transformations, contains $n$ reflections and $n$ rotations. For a given $n$, the elements $\mathcal{T}^s_k\in\mathbb{T}_n$ are given by,
\begin{align}
\mathcal{T}_k^s:=\begin{pmatrix}
\cos\frac{2\pi k}{n} & -s\sin\frac{2\pi k}{n} & 0 \\
\sin\frac{2\pi k}{n} & s\cos\frac{2\pi k}{n} & 0 \\
0 & 0 & 1   
\end{pmatrix},
\end{align}
where $k\in\{1,2,\cdots,n\}$ and $s=\pm1$, with $s=+$ corresponding to the rotations and $s=-$ representing the reflections. The composition of such models will be discussed in the next section. In the following subsection, we recall another prerequisite concept, Hardy's nonlocality argument. 

\subsection{Hardy's nonlocality test}
Hardy's argument can be considered as one of the simplest tests of Bell nonlocality for two spacelike separated parties, Alice and Bob, sharing parts of a composite system. In each run, Alice and Bob perform  either of the dichotomic measurements $\mathcal{M}_x$, $\mathcal{M}_y$ or $\mathcal{N}_x$, $\mathcal{N}_y$ on their part of the composite system respectively, where $x,y\in\{1,2\}$. Their local outcomes are denoted as $a$ and $b$ respectively with $a,b\in\{\pm 1\}$. Upon performing the experiment many times they will obtain a joint input-output probability distribution, also called behaviour, $P\equiv\{p(a,b|\mathcal{M}_x^A,\mathcal{N}_y^B)~|~x,y\in\{1,2\};~a,b=\pm1\}$. A behavior is termed local, if it can be expressed in factorized form, {\it i.e.}, $p(a,b|\mathcal{M}_x^A,\mathcal{N}_y^B)=\int_\Lambda d\lambda p(\lambda)\xi(a|\mathcal{M}_x^A,\lambda)\xi(b|\mathcal{N}_y^B,\lambda)$, where $p(\lambda)$ is a probability distribution over a set of local hidden-variables $\Lambda$, and $\xi(a|\mathcal{M}_x^A,\lambda),~\xi(b|\mathcal{N}_y^B,\lambda)$ are local response functions of Alice and Bob respectively. As shown by Hardy, a correlation satisfying the following four conditions
\begin{subequations}
\begin{align}
\label{hardysucc}
p\left(+1,+1|\mathcal{M}_1^A,\mathcal{N}_1^B\right)&>0,\\
P(+1,+1|\mathcal{M}_1^A,\mathcal{N}_2^B)&=0,\\
P(+1,+1|\mathcal{M}_2^A,\mathcal{N}_1^B)&=0,\\
P(-1,-1|\mathcal{M}_2^A,\mathcal{N}_2^B)&=0,
\end{align}
\end{subequations}
must be nonlocal. The quantity on the left hand side of Eq.(\ref{hardysucc}) is known as the success probability ($\mathrm{P}_H$) of Hardy's argument, {\it i.e.} $\mathrm{P}_H:=p\left(+1,+1|\mathcal{M}_1^A,\mathcal{N}_2^B\right)$. In quantum mechanics the maximum value of Hardy's success probability has been shown to be $(-11 + 5\sqrt{5})/2\approx 0.09$ \cite{Seshadreesan11,Rabelo12}. Apart from certifying nonlocality, Hardy's test has recently been shown to have several applications, such as post quantumness witness \cite{Das13(1),Das13(2),Das13(3)}, randomness certification \cite{Ramanathan18}, and self-testing quantum states \cite{Rai21(1),Rai21(2)}. In this work, we will utilize this to study the correlation strength of bipartite polygon systems.  

\section{Entanglement classes in bipartite polygons}\label{three}
As already mentioned, the bipartite composition for any GPT system should lie in between the minimal and the maximal tensor product compositions. For polygon systems, any such composition must contain all the product states and product effects. The sets of such extremal states and ray extremal effects are given by 
\begin{subequations}
\begin{align}
\mathcal{P}_{st}[n]&\equiv\left\{\omega_{n(i-1)+j}:=\omega_i^A\otimes\omega_j^B\right\}_{i,j=1}^n,\label{extsta}\\
\mathcal{P}_{ef}[n]&\equiv\left\{e_{n(i-1)+j}:=e_i^A\otimes e_j^B\right\}_{i,j=1}^n,\label{rayexteff}
\end{align}
\end{subequations}
where $\omega^A_i,\omega^B_j$ are defined in Eq.(\ref{state}) and $e^A_i,e^B_j$ are defined in Eq.(\ref{effect}). Instead of representing a bipartite state/effect as a vector in $\mathbb{R}^9$, it is sometimes convenient to represent it as a $3\times 3$ matrix, {\it i.e.} $\omega_i\otimes\omega_j\equiv\omega_i\otimes\omega_j^{\mathrm{T}}$. Apart from the aforesaid product or factorized states, a bipartite composition may also allow states that are not factorized, and such states are called entangled. For $n$-gon one such entangled state is identified in \cite{Janotta11}:
\begin{align}
\Phi_J:=\begin{pmatrix}1&0&0\\0&1&0\\0&0&1\end{pmatrix},~~\mbox{for~odd~} n;~~~~~~~~~~~~~~~\label{oddJ}\\
\Phi_J:=\begin{pmatrix}\cos(\pi/n)&\sin(\pi/n)&0\\-\sin(\pi/n)&\cos(\pi/n)&0\\0&0&1\end{pmatrix},~~\mbox{for~even~} n.\label{evenJ}
\end{align}
As pointed out in \cite{Janotta11}, the state $\Phi_J$ can be considered as a natural analogue of quantum mechanical maximally entangled state. Here, we use the sub-index following the initial of the first author of Ref.\cite{Janotta11}. Note that any such entangled state must yield consistent probability on any product effect.   

Like the entangled states, there exist entangled effects as well in a composite theory, which must yield consistent probability on any product state. Whenever a composite system is considered to allow both entangled states and entangled effects, then it must be taken care that all the effects yield consistent probabilities on all the states. For the composition of two square bits, it has been shown that there exist exactly four types of compositions \cite{DallArno17}. Among them, one is the maximal composition that allows all possible factorized and entangled states but admits only factorized effects. One is the other extreme, {\it i.e.} minimal composition that permits only factorized states, but all possible factorized and entangled effects are allowed. The other two compositions lie strictly in between. In this work, we will be interested in the maximal composition of two polygon systems. For such a bipartite composition, the set of allowed reversible transformations is given by  
\begin{align}
\mathbb{T}_{AB}[n]&:=\left\{\mathcal{T}^{s_1}_{k_1}\otimes \mathcal{T}^{s_2}_{k_2},~\mbox{Swap}~|~\mathcal{T}^{s_1}_{k_1}\in\mathbb{T}^A_n,~\mathcal{T}^{s_2}_{k_2}\in\mathbb{T}^B_n\right\}\nonumber,\\
&\mbox{where},~\mbox{Swap}(\omega_i^A\otimes\omega_j^B):=\omega_j^A\otimes\omega_i^B;\nonumber\\
&s_1,s_2\in\{\pm1\};~k_1,k_2\in\{1,2,\cdots,n\}.\label{unitaries} 
\end{align}
Note that the Swap operation, as well as the local reversible transformations, map product states to product ones and entangled to entangled ones. These help us to define equivalent classes of entangled states. 
\begin{definition}\label{defi1}
Two entangled states $\Phi_1$ and $\Phi_2$ of a bipartite polygonal system are said to be equivalent if they are connected by some local reversible transformation.   
\end{definition}
For instance, the maximal composition of two square bits allows $8$ different entangled states \cite{DallArno17}. However, all of them belong to the same class as they are connected with each other by local reversible transformations. However, in the following, we will see that this is not the case for higher gons.

\subsection{Finding extremal entangled states}\label{IIIA}
Any unnormalized state in the maximal composition of the bipartite polygon system can be represented as a vector in $\mathbb{R}^{9}$, which can also be represented as a 3$\times$3 matrix, 
\begin{align}
\Phi\equiv\begin{pmatrix}
a_1 & a_2 & a_3 \\
a_4 & a_5 & a_6 \\
a_7 & a_8 & a_9   
\end{pmatrix},~~~~a_i\in\mathbb{R}.
\end{align}
Positivity of outcome probability demands $\tr[e^{\mathrm{T}}(\Phi)]\geq 0$ for any effect allowed in this theory. Note that in maximal composition only product effects are allowed which forms an effect cone in $\mathbb{R}^9$, with the ray extremal effects $\mathcal{P}_{ef}[n]$ as defined in Eq.(\ref{rayexteff}). Therefore, the required positivity is assured once it is checked that $\tr[e^{T}(\Phi)]\geq 0,~ \forall~e\in\mathcal{P}_{ef}[n]$.  

Normalization of the state is defined with the help of unit effect $u:=u_A\otimes u_B^{\mathrm{T}}$, which demands $\tr[u^{\mathrm{T}}(\Phi)]=1$, and accordingly we have $a_9=1$. Thus a normalized state reads as
\begin{align}
\Phi\equiv\begin{pmatrix}
a_1 & a_2 & a_3 \\
a_4 & a_5 & a_6 \\
a_7 & a_8 & 1   
\end{pmatrix},
\end{align}
and the set of normalized states forms a $8$-dimensional polytope embedded in $\mathbb{R}^9$. 
\begin{figure}[b!]
\centering
\includegraphics[width=0.55\textwidth]{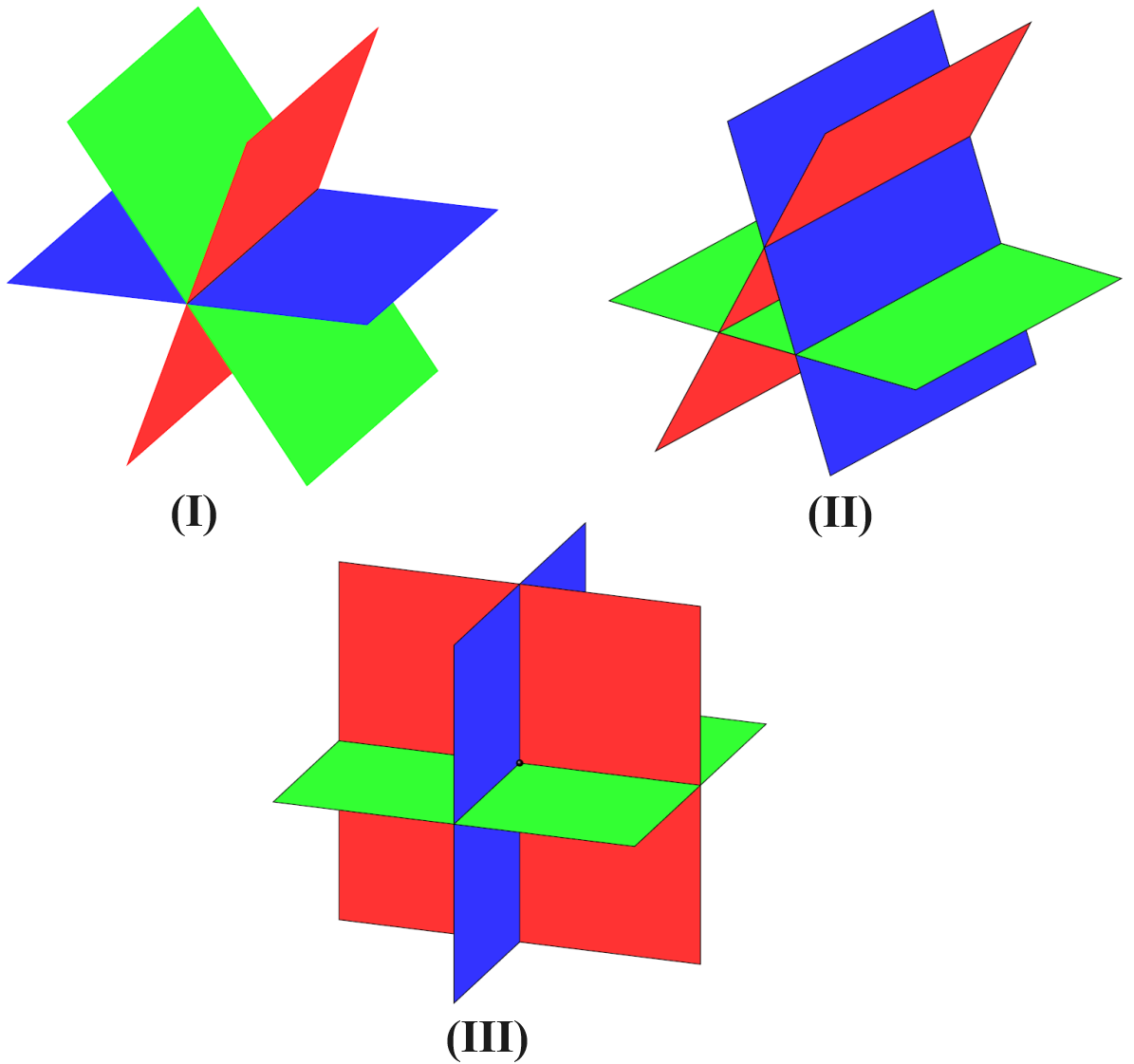}
\caption{Three mutually non-parallel planes in $\mathbb{R}^3$ can intersect each other in three different ways. While in case (I) all three planes intersect in a common line, in case (II) each of the pairs intersect in different lines. On the other hand, in case (III) they intersect in a common unique point, which is of our interest.}\label{fig1}
\end{figure}

To find the extreme points of the normalized state space, first note that the inequality $\tr[e(\Phi)]\ge0$ represents an $8$-dimensional half-space for any normalized state $\Phi$ and for any effect $e\in\mathcal{P}_{ef}[n]$. If we represent this expression as equality, {\it i.e.}, $\tr[e(\Phi)]=0$, then it corresponds to a $7$-dimensional hyper-surface. Now for a different $e^\prime\in\mathcal{P}_{ef}[n]$ the equation $\tr[e^\prime(\Phi)]=0$ corresponds to another $7$-dimensional hyper-surface which is either parallel to the hyper-surface corresponding to the effect $e$ or they intersect each other in a $6$-dimensional hyper-surface. Now for a third effect $e^{\prime\prime}$, if the $7$-dimensional hyper-surface $\tr[e^{\prime\prime}(\Phi)]=0$ is not parallel to the hyper-surfaces corresponding to $e$ and $e^\prime$ then it may intersect them in the same $6$-dimensional hyper-surface (the intersection of $\tr[e(\Phi)]=0$ and $\tr[e^\prime(\Phi)]=0$ hyper-surfaces) or these three intersect each other in a $5$-dimensional hyper-surface (see Fig.\ref{fig1}). By proceeding this way, at some stage, $8$ different hyper-surface corresponding to eight different effects from $\mathcal{P}_{ef}[n]$ will intersect at a single point, which corresponds to a normalized pure state. Mathematically this boils down to checking the uniqueness of the solution for a system of linear equations. The eight different effects from the set $\mathcal{P}_{ef}[n]$ can be chosen in $^{n^2}C_8$ different ways. All of these choices will not lead to a unique solution, but whenever it does, we obtain an extremal bipartite state for the maximal composition of the polygonal systems, provided the positivity constraints are satisfied. See the Appendix for a more detailed discussion.  

Once an extreme state $\Phi$ is identified, it is then straightforward to check whether it belongs to the set $\mathcal{P}_{st}[n]$ or not. If it does not belong to the set $\mathcal{P}_{st}[n]$, it corresponds to an extreme entangled state. Furthermore, the entangled states can be classified (see Definition \ref{defi1}) with the help of local reversible operations chosen from the set $\mathcal{T}_{AB}[n]$.    
\subsection{Bipartite pentagon system} 
Following the above procedure in MATLAB, for bipartite composition of pentagon systems we obtain $135$ different extreme states $\Phi_k$; $k\in\{1,2,\cdots,135\}$. Among these we have $25$ (say, number $1$ to $25$) factorized extreme states $\Phi_{5(i-1)+j} = \omega_i\otimes\omega_j^{\mathrm{T}}$, where $i,j \in \{1,\cdots,5\}$ and $\omega_i$'s are given in Eq.(\ref{state}). The remaining $110$ states (number $26$ to $135$) are entangled. Furthermore, applying the local reversible transformation we find that these states can be divided into two classes -- (i) the first class contains $10$ states (say, number $26$ to $35$), and (ii) the second one contains $100$ states (number $36$ to $135$). One representation state of the first class is the state $\Phi_J$ of Eq.(\ref{oddJ}) and hence we call this {\it Janotta} class. One of the representation states for the second class is given by 
\begin{align}
\Phi_H=\begin{pmatrix}
-\cos(\pi/5) & \frac{-r_{5}^{6} \sin(\pi/5)}{8\left(1 + r_{5}^{2}\right)} & 0 \\
\frac{-r_{5}^{6} \sin(\pi/5)}{8\left(1 + r_{5}^{2}\right)} & \cos(\pi/5) & \frac{-r_{5}^{3}}{4 \sin(\pi/5)}\\
0 & \frac{-r_{5}^{3}}{4 \sin(\pi/5)} & 1 \\
\end{pmatrix}.\label{pentaH}
\end{align}   
This class we call the {\it Hardy} class, so the sub-index `H'. The justification of this nomenclature will become obvious in the next section. The other states in {\it Janotta} class and {\it Hardy} class can be obtained from the representative states $\Phi_J$ and $\Phi_H$ respectively by applying local reversible transformations.

The two classes of states $\Phi_J$ and $\Phi_H$ have structural distinctions. For instance the state in Eq.(\ref{oddJ}) we have $a_3=a_6=a_7=a_8=0$. It can be shown that all the states in this class (obtained through local reversible transformations) have the same feature, which is not the case for the states in the class of $\Phi_H$.
\begin{figure}[t]
\centering
\includegraphics[width=0.65\textwidth]{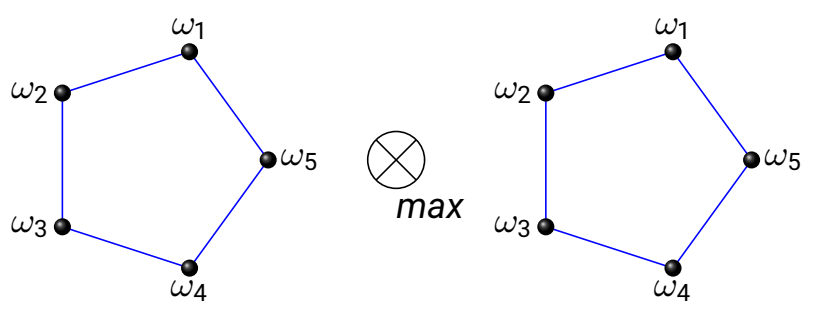}
\caption{Maximal tensor product of two elementary pentagon systems allows $25$ product states. On the other hand, it allows two different classes (not equivalent under local reversible transformation) of entangled states: Janotta class states with $\Phi_J$ of Eq.(\ref{oddJ}) being a representative state and the Hardy class states with $\Phi_H$ of Eq.(\ref{pentaH}). While $\Phi_J$ can be thought of as a natural analog of the maximally entangled state of a two-qubit and does not exhibit Hardy's nonlocality, the $\Phi_H$ state shows Hardy and importantly with the success probability strictly greater than quantum success. However, the resulting correlation belongs to the set of {\it almost quantum set } $Q_1$.}\label{fig2}
\end{figure}
\begin{figure}[b]
\centering
\includegraphics[width=0.65\textwidth]{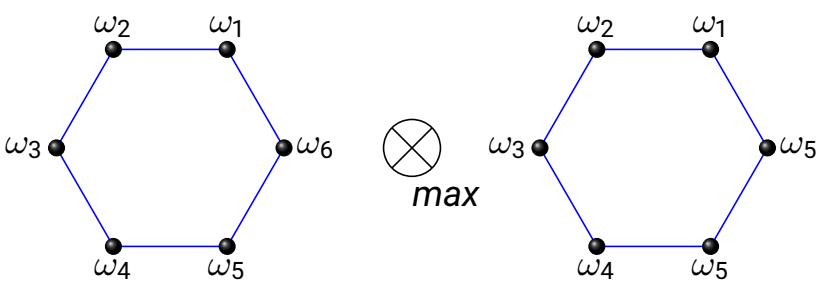}
\caption{Maximal tensor product of two elementary hexagon systems allows $36$ product state. On the other hand, it allows six different classes of entangled states. Representative states for each of the classes are given in Eq.(\ref{hexa}). The state $\Phi_I$ although can be thought of as an analog of the maximally entangled state, unlike the two-qubit maximally entangled state, exhibits Hardy's nonlocality behaviour.}\label{fig3}
\end{figure}

\subsection{Bipartite hexagon system}
The task of characterizing all the entangled states becomes computationally costly with higher gons. This is because the choices of eight different effects rapidly increase with the number of sides in the component polygons. However, we obtain a complete characterization of the entanglement states for the bipartite hexagon. It turns out that there are six different entangled classes of states possible there. Representation states for each of these classes are given below:
%\footnotesize
\begin{align}\label{hexa}
\Phi_I&=\begin{pmatrix}
\frac{\sqrt{3}}{2} & \frac{1}{2} & 0\\
\frac{1}{2} & -\frac{\sqrt{3}}{2} & 0\\
0 & 0 & 1
\end{pmatrix},~~~
\Phi_{II}=\begin{pmatrix}
\frac{1}{15r_6^2} & \frac{7}{10} & \frac{2}{3r_6^3}\\
\frac{7}{10} & \frac{1}{5r_6^2} & \frac{3}{5r_6}\\
\frac{2}{3r_6^3} & \frac{3}{5r_6} & 1
\end{pmatrix},\nonumber\\
\Phi_{III}&=\begin{pmatrix}
\frac{1}{3r_6^2} & \frac{1}{r_6^4} & \frac{2}{3r_6^3}\\
\frac{1}{r_6^4} & 0 & \frac{1}{2r_6}\\
\frac{2}{3r_6^3} & \frac{1}{2r_6} & 1
\end{pmatrix},~
\Phi_{IV}=\begin{pmatrix}
\frac{1}{7r_6^2} & \frac{9}{14} & \frac{10}{21r_6^3}\\
\frac{11}{14} & \frac{1}{7r_6^2} & \frac{5}{7r_6}\\
\frac{6}{7r_6^3} & \frac{3}{7r_6} & 1
\end{pmatrix},\\
\Phi_{V}&=\begin{pmatrix}
\frac{1}{7r_6^2} & \frac{11}{14} & \frac{6}{7r_6^3}\\
\frac{9}{14} & \frac{1}{7r_6^2} & \frac{3}{7r_6}\\
\frac{10}{21r_6^3} & \frac{5}{7r_6} & 1
\end{pmatrix},~
\Phi_{VI}=\begin{pmatrix}
\frac{-1}{2r_6^2} & \frac{-1}{r_6^4} & \frac{1}{3r_6^3}\\
\frac{-1}{r_6^4} & \frac{1}{2r_6^2} & \frac{-1}{2r_6}\\
\frac{1}{3r_6^3} & \frac{-1}{2r_6} & 1
\end{pmatrix}.\nonumber
\end{align}
The state $\Phi_I$ is equivalent to the maximally entangled state $\Phi_J$ of Eq. (\ref{evenJ}) as identified by Janotta {\it et al.} Important to note that the states $\Phi_{I},\Phi_{II},~\Phi_{III}$ and $\Phi_{VI}$ are the symmetric representatives of their corresponding class, while we cannot find any symmetric state in class $\Phi_{IV}$ and $\Phi_{V}$. Recall that, a joint state $\Phi^{AB}$ is called symmetric if $(e\otimes f)(\Phi^{AB})=(f\otimes e)(\Phi^{AB}),~\forall~e,f\in V^\star_+$, and in matrix representation $\Phi^{AB}$ is symmetric if and only if the corresponding matrix is symmetric \cite{Janotta11}. However, we can notice that the states  $\Phi_{IV}$ and $\Phi_{V}$ are transposes of each other meaning that even though they are not connected by local reversible dynamics they are related by a swap operation. Thus the entanglement content in these two states is the same.

\section{Nonlocal properties of the entangled states}\label{four}
Nonlocal properties of the correlations obtained from the maximally entangled states $\Phi_J$ have been studied in Ref. \cite{Janotta11}. In particular, the maximal CHSH inequality violation has been explored for even and odd gons. Importantly, the correlations of even $n$ systems can always reach or exceed Tsirelson’s bound, while the correlations of odd $n$ systems are always below Tsirelson’s bound.  For the odd $n$ systems the maximally entangled state belongs to the class {\it inner product states}\footnote{A state $\Phi^{AB}$ is called an inner product state if $\Phi^{AB}$ is symmetric, and positive
semi-definite, {\it i.e.} $(e\otimes e)(\Phi^{AB})\ge0,~\forall~e\in V^\star$ \cite{Janotta11}.} and all correlations obtainable from measurements on inner product states satisfy Tsirelson’s bound. Here we analyze the nonlocal properties of different classes of entangled states from the perspective of Hardy's nonlocality argument.  

As already mentioned we need two dichotomic measurements both on Alice's and Bob's part to construct Hardy's non-locality argument. Furthermore, Alice (also Bob) must choose incompatible measurements on her (his) subsystem \cite{Fine82(1),Fine82(2)}. Let Alice's first measurement is $\mathcal{M}_1^A=\left\{e_i^A,\bar{e}_i^A\right\}\equiv\left\{E^{A+}_1,E^{A-}_1\right\}$. We use the notation $\mathcal{M}_1^A=\left\{E^{A+}_1,E^{A-}_1\right\}$ to give the freedom that positive outcome can be assigned to $e^A_i$ or $\bar{e}^A_i$. Let the other measurement of Alice be $\mathcal{M}_2^A=\left\{e_j^A,\bar{e}_j^A\right\}\equiv\left\{E^{A+}_2,E^{A-}_2\right\}$, with $j\neq i$. Similarly, for Bob let us consider two measurements $\mathcal{N}_1^B=\left\{e_k^B,\bar{e}_k^B\right\}\equiv\left\{E^{B+}_1,E^{B-}_1\right\}$ and $\mathcal{N}_2^B=\left\{e_l^B,\bar{e}_l^B\right\}\equiv\left\{E^{B+}_2,E^{B-}_2\right\}$, with $l\neq k$. With these measurement choices, up to local relabeling, Hardy's nonlocality argument can be written as
\begin{subequations}
\label{h}
\begin{align}
P(E^{A+}_1,E^{B+}_1|\mathcal{M}_1^A\mathcal{N}_1^B)>0\label{h1},\\
P(E^{A+}_1,E^{B+}_2|\mathcal{M}_1^A\mathcal{N}_2^B)=0\label{h2},\\
P(E^{A+}_2,E^{B+}_1|\mathcal{M}_2^A\mathcal{N}_1^B)=0\label{h3},\\
P(E^{A-}_2,E^{B-}_2|\mathcal{M}_2^A\mathcal{N}_2^B)=0\label{h4}.
\end{align}
\end{subequations}

\subsection{Hardy's nonlocality for maximally entangled states}
In this subsection, we will analyze Hardy's nonlocality behaviour of the correlations obtained from the maximally entangled states of bipartite polygon theories. We prove two generic results. In the following, we first prove a no-go result.  
\begin{figure}[b!]
\centering
\includegraphics[width=0.45\textwidth]{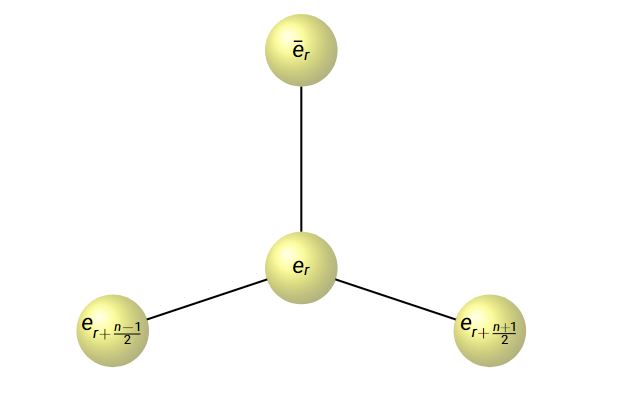}
\caption{(Color online) Orthogonality graph of extreme effects for odd-gon theories. Each node denotes an extreme effect. Two effects $e$ and $f$ are connected with each other by an edge if and only if they are orthogonal to each other in the sense that $e\cdot f=0$; here, the inner product is standard $\mathbb{R}^3$ inner product. Operationally, for a pair of such effects there always exist at least one pair of states $\omega_e$ and $\omega_f$, such that $e\cdot\omega_e=1=f\cdot\omega_f$ and $e\cdot\omega_f=0=f\cdot\omega_e$. While calculating the sub-indices of the effects modulo $n$ operation is assumed throughout. Here, in order to be consistent with our notation we define `$r~\mbox{mod}~n$' in such a way that it returns the remainder if the remainder is nonzero, otherwise it returns $n$.}\label{ortho}
\end{figure}

%\footnote{} 

\begin{theorem}\label{the1}
The maximally entangled states $\Phi_J$ of the bipartite regular polygons with odd $n$ do not exhibit Hardy's nonlocality argument.
\end{theorem}
\begin{proof}
The outcome probability for any pair of effect $E$ on Alice's side and $F$ on Bob's side for the maximally entangled state $\Phi_J$ of Eq.(\ref{oddJ}) read as
\begin{align}
P_{\Phi_J}(E,F)&=\tr\left[\left(E\otimes F^\mathrm{T}\right)\Phi_J^\mathrm{T}\right] \nonumber\\
&=\tr\left[E\otimes F^\mathrm{T}\right]=E\cdot F,
\end{align}
where $E\cdot F$ be the usual inner product in $\mathbb{R}^3$. With this, the Hardy conditions of Eq.(\ref{h}) become
\begin{subequations}
\label{ho}
\begin{align}
P(E^{A+}_1,E^{B+}_1|\mathcal{M}_1^A\mathcal{N}_1^B)=E^{A+}_1\cdot E^{B+}_1>0\label{ho1},\\
P(E^{A+}_1,E^{B+}_2|\mathcal{M}_1^A\mathcal{N}_2^B)=E^{A+}_1\cdot E^{B+}_2=0\label{ho2},\\
P(E^{A+}_2,E^{B+}_1|\mathcal{M}_2^A\mathcal{N}_1^B)=E^{A+}_2\cdot E^{B+}_1=0\label{ho3},\\
P(E^{A-}_2,E^{B-}_2|\mathcal{M}_2^A\mathcal{N}_2^B)=E^{A-}_2\cdot E^{B-}_2=0\label{ho4}.
\end{align}
\end{subequations}
To see whether the aforesaid conditions can be satisfied in any odd-gon, it is handy to have a look at the orthogonality graph for the extreme effects of the odd-gon theory. Two effects $e$ and $f$ will be called orthogonal to each other if and only if $e\cdot f=0$. In any odd-gon theory, it turns out that a ray extremal effect $e_r$ is orthogonal to only two other ray extremely effects $e_{s^\pm}$, with $s^\pm:=\left(r+\frac{n\pm1}{2}\right)~\mbox{mod}~n$; and to the (non-ray) extremal effect $\bar{e}_r$. The orthogonality graph is shown in Fig. \ref{ortho}. The proof of the theorem follows by analyzing the following two cases. 

{\it Case-I:} Let us assume that the effect $E^{B+}_2$ in Eq.(\ref{ho2}) corresponds to some ray extremal effect (say) $e_r$. This implies $E^{B-}_2=\bar{e}_r$, since $\{E^{B+}_2,E^{B-}_2\}$ correspond to a measurement. Eq.(\ref{ho4}) further implies that $E^{A-}_2$ must be orthogonal to $E^{B-}_2$.
Since the only extreme effect orthogonal to $\bar{e}_r$ is $e_r$ (see Fig. \ref{ortho}), therefore we must have $E^{A-}_2=e_r$, which further implies $E^{A+}_2=\bar{e}_r$. Again Eq.(\ref{ho3}) implies $E^{B+}_1=e_r$ and hence $E^{A+}_1\cdot E^{B+}_1=E^{A+}_1\cdot E^{B+}_2=0$, implying zero Hardy success.\\

{\it Case-II:} Here we start by assuming that the effect $E^{B+}_2$ in Eq.(\ref{ho2}) corresponds to some non ray extremal effect (say) $\bar{e}_r$. This implies $E^{B-}_2=e_r$, since $\{E^{B+}_2,E^{B-}_2\}$ correspond to a measurement. This also implies $E^{A+}_1=e_r$ from Eq.(\ref{ho2}). Now from Eq.(\ref{ho4}) we know that $E^{A-}_2$ is orthogonal to $E^{B-}_2$. Which implies either $E^{A-}_2=e_{s^\pm}$ $\left(\mbox{with}~ s^\pm:=\left(r+\frac{n\pm1}{2}\right)~\mbox{mod}~n\right)$ or $E^{A-}_2=\bar{e}_r$. If $E^{A-}_2=e_{s^\pm}$ we have  $E^{A+}_2=\bar{e}_{s^\pm}$ since $\{E^{A+}_2,E^{A-}_2\}$ forms a measurement. Which from Eq.(\ref{ho3}) further implies $E^{B+}_1=e_{s^\pm}$, yielding $E^{A+}_1\cdot E^{B+}_1=e_r\cdot e_{s^\pm}=0$. Similarly if $E^{A-}_2=\bar{e}_r$ we would get $E^{A+}_2=e_r$ and $E^{A+}_1\cdot E^{B+}_1=E^{A+}_2\cdot E^{B+}_1=0$. This proves that even for {\it Case-II} we have zero Hardy success probability, and hence this completes the proof.
\end{proof}
While maximally entangled states of bipartite odd gons do not exhibit Hardy's nonlocality, maximally entangled states of bipartite square bit do exhibit such nonlocality. The PR box correlation resulting from the maximally entangled state of the bipartite square bit exhibits Hardy's nonlocality argument with success probability $1/2$, which is, in fact, the maximum Hardy's success among any no-signaling correlations. In the following, we prove a generic result that the maximally entangled state of any bipartite even gon depicts Hardy's type of nonlocality, albeit with decreasing success probability.

At this point, it is worth mentioning that the maximally entangled state of the quantum two-qubit system fails to exhibit Hardy's nonlocality behaviour \cite{Goldstein94}. In this sense, odd gons are closer to quantum than even gons as the maximally entangled states of the former do not depict Hardy nonlocality while the latter do.
\begin{theorem}\label{theo2}
The maximally entangled state $\Phi_J$ of bipartite even-gons (with $n\ge 4$) exhibits Hardy's non-locality argument with the success probability given by $\sin^2\frac{\pi}{n}$. 
\end{theorem}
\begin{proof}
For even-gons, all the extreme effects are ray extremal. The outcome probability $P_{\Phi_J}\left(e_i^A,e_j^B\right)$ of Alice's effect $e_i^A$ and Bob's effect $e_j^B$ on the maximally entangled state $\Phi_J$ of Eq.(\ref{evenJ}) reads as
\begin{align*}
P_{\Phi_J}\left(e_i^A,e_j^B\right)&=\Tr\left[\left(e_i\otimes e_j^\mathrm{T}\right)\Phi_J\right]\nonumber\\
&=\frac{1}{4}\left[\sec\left(\frac{\pi}{n}\right)\cos\left(\frac{2(i-j)\pi}{n}-\frac{\pi}{n}\right)\right]+\frac{1}{4}.
\end{align*}
Now $P_{\Phi_J}\left(e_i^A,e_j^B\right)=0$ if and only if 
\begin{align}
\cos\left(\frac{2(i-j)\pi}{n}-\frac{\pi}{n}\right)&=-\cos\left(\frac{\pi}{n}\right)\nonumber\\
i.e.,~ i-j&=\frac{1}{2}+\frac{n\pm 1}{2},\label{i-j}
\end{align}
where in the last expression modulo $n$ addition is implied. Let us consider $E^{A+}_1$ in Eq.(\ref{h}) be some ray extremal effect $e_r$ for some $r\in\{1,2,\cdots,n\}$. Eq.(\ref{i-j}) implies that to satisfy the condition of Eq.(\ref{h2}), we must have $E^{B+}_2=e_{s(\alpha)}$, where $s(\alpha)=r-\frac{(n+1)}{2}+(-1)^\alpha \frac{1}{2}$, with Greek indices taking values from $\{0,1\}$. Since $E^{B+}_2$ and $E^{B-}_2$ forms a measurement, therefore we have $E^{B-}_2=e_{t(\alpha)}$, where $t(\alpha)=r-\frac{1}{2}+(-1)^\alpha\frac{1}{2}$. Again, Eqs.(\ref{i-j}) and (\ref{h4}) imply $E^{A-}_2=e_{v(\alpha,\beta)}$, where $v(\alpha,\beta)=r+[(-1)^\alpha+(-1)^\beta ]\frac{1}{2}+\frac{n}{2}$, and $E^{A+}_2=e_{w(\alpha,\beta)}$, with $w(\alpha,\beta)=r+[(-1)^\alpha+(-1)^\beta ]\frac{1}{2}$. Finally, Eqs.(\ref{i-j}) and (\ref{h3}) yield $E^{B+}_1:=e_{z(\alpha,\beta,\gamma)}=r+[(-1)^\alpha+(-1)^\beta+(-1)^\gamma] \frac{1}{2}-\frac{1}{2}-\frac{n}{2}$. In other words, given the choice $E^{A+}_1=e_r$ the effect $E^{B+}_1$ has only the following four choices  
\begin{subequations}
\begin{align}
E^{B+}_1&= e_{z(\alpha)};~z(\alpha):=r-\frac{(n+1)}{2}+(-1)^\alpha \frac{1}{2},\label{E1}\\
E^{B+}_1&= e_{z(\alpha)};~z(\alpha):=r-\frac{(n+1)}{2}+(-1)^\alpha \frac{3}{2}.\label{E2}
\end{align}
\end{subequations}
For the case of Eq.(\ref{E1}), we have
\begin{align*}
&P_{\Phi_J}\left(E^{A+}_1,E^{B+}_1\right)\\
&=\frac{1}{4}\left[\sec\left(\frac{\pi}{n}\right)\cos\left(\frac{2(r-z(\alpha))\pi}{n}-\frac{\pi}{n}\right)\right]+\frac{1}{4}\\
&=\frac{1}{4}\left[\sec\left(\frac{\pi}{n}\right)\cos\left(\frac{2(\frac{n}{2}+\frac{1}{2} \pm  \frac{1}{2})\pi}{n}-\frac{\pi}{n}\right)\right]+\frac{1}{4}\\
&=\frac{1}{4}\left[\sec\left(\frac{\pi}{n}\right)\cos\left(\pi \pm \frac{\pi}{n}\right)\right]+\frac{1}{4}\\
&=\frac{1}{4}\left[-\sec\left(\frac{\pi}{n}\right)\cos\left( \frac{\pi}{n}\right)\right]+\frac{1}{4}=0.
\end{align*}
Therefore, these particular choices of $E^{B+}_1$ do not exhibit Hardy nonlocality. However, for the choices of Eq.(\ref{E2}) we obtain  
\begin{align*}
&P_{\Phi_J}\left(E^{A+}_1,E^{B+}_1\right)\\
&=\frac{1}{4}\left[\sec\left(\frac{\pi}{n}\right)\cos\left(\frac{2(r-z(\alpha))\pi}{n}-\frac{\pi}{n}\right)\right]+\frac{1}{4}\\
&=\frac{1}{4}\left[\sec\left(\frac{\pi}{n}\right)\cos\left(\frac{2(\frac{n}{2}+\frac{1}{2} \pm  \frac{3}{2})\pi}{n}-\frac{\pi}{n}\right)\right]+\frac{1}{4}\\
&=\frac{1}{4}\left[\sec\left(\frac{\pi}{n}\right)\cos\left(\pi \pm \frac{3\pi}{n}\right)\right]+\frac{1}{4}\\
&=\frac{1}{4}\left[-\sec\left(\frac{\pi}{n}\right)\cos\left( \frac{3\pi}{n}\right)\right]+\frac{1}{4}=\sin^2\left(\frac{\pi}{n}\right).
\end{align*}
This completes the proof of the theorem. The variation of Hardy's success probability for different even gons is shown in Fig. \ref{fig3}. 
\end{proof}
\begin{figure}[t!]
\centering
\includegraphics[width=0.5\textwidth]{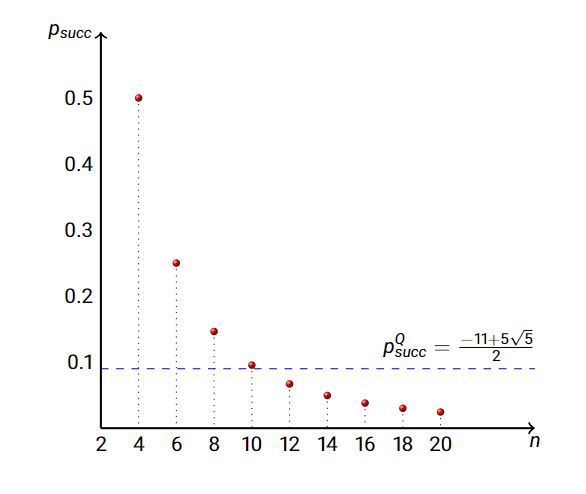}
\caption{(Color online) Red dots denote the maximum success probability of Hardy's nonlocality argument for maximally entangled states of bipartite even gons. The Blue dashed line denotes the optimal quantum success probability of Hardy's nonlocality argument which, in contrast, is obtained for the \textit{non-maximally entangled state}.}\label{fig3}
\end{figure}

\subsection{Hardy's nonlocality for non-maximally entangled states}
For the bipartite pentagon case, we only have two inequivalent classes of pure entangled states, the state $\Phi_J$ of Eq.(\ref{oddJ}) and the state $\Phi_H$ of Eq.(\ref{pentaH}). As already established in Theorem \ref{the1}, the state $\Phi_J$ cannot result in any correlation exhibiting Hardy's nonlocality. So the natural question arises whether the state $\Phi_H$ can lead to such a correlation. Interestingly we find that that state $\Phi_H$ indeed exhibits Hardy's nonlocality. If we consider two incompatible measurements $\mathcal{M}_1\equiv\{e_1,\bar{e}_1\}$ and $\mathcal{M}_2\equiv\{e_5,\bar{e}_5\}$ on Alice's part and two incompatible measurements $\mathcal{N}_1\equiv\{e_1,\bar{e}_1\}$ and $\mathcal{N}_2\equiv\{\bar{e}_2,e_2\}$ on Bob's part, the resulting correlation depicts Hardy's nonlocality. Denoting the outcome corresponding to the first effect as $+1$ and the outcome corresponding to the second one as $-1$, the correlation obtained from these choices of measurements reads as
{\small \begin{align}
\mathbf{P}	
~\equiv~
\begingroup
\setlength{\tabcolsep}{2pt} % Default value: 6pt
\renewcommand{\arraystretch}{2} % Default value: 1
\begin{array}{c||c|c|c|c|} 
 &(+,+) & (+,-)&(-,+)&(-,-)\\ \hline\hline
\mathcal{M}_{1}\mathcal{N}_{1}&1 -\frac{4\sqrt{5}}{10} & \frac{7\sqrt{5}}{10} - \frac{3}{2} & \frac{7\sqrt{5}}{10} - \frac{3}{2} & 3 - \sqrt{5}\\\hline 
\mathcal{M}_{1}\mathcal{N}_{2} &0 & \frac{3\sqrt{5}}{10} - \frac{1}{2} & \frac{\sqrt{5}}{10} + \frac{1}{2} & 1 -\frac{4\sqrt{5}}{10}\\\hline
\mathcal{M}_{2}\mathcal{N}_{1} &0 & \frac{2\sqrt{5}}{10} & \frac{3\sqrt{5}}{10} - \frac{1}{2}& \frac{3}{2}- \frac{5\sqrt{5}}{10}\\\hline
\mathcal{M}_{2}\mathcal{N}_{2} &\frac{3\sqrt{5}}{10} - \frac{1}{2} & \frac{1}{2} - \frac{\sqrt{5}}{10} & 1 -\frac{2\sqrt{5}}{10} & 0\\\hline
\end{array}~.
\endgroup
\label{pentaco}
\end{align}}
Important to note that the success probability of Hardy's argument in this case is $P\left(+,+|\mathcal{M}_1\mathcal{N}_1\right)=1-\frac{4\sqrt{5}}{10}\approx0.1056$, which is strictly larger than the corresponding optimal quantum value $\frac{5\sqrt{5}-11}{2}\approx0.0902$. Therefore, this particular correlation is beyond quantum in nature, although its CHSH value is strictly less than the Cirel'son's value. The success probability $0.1056$ turns out to be optimal in pentagon theory. 
For the hexagon case, we find that all the six different classes of states depicts Hardy's nonlocality. The choices of measurements and the corresponding Hardy's success probabilities are listed in Table \ref{tab1}.
\begin{table}[t!]
\centering
\begin{tabular}{|c|c|c|c|}
\hline
\multirow{2}{3em}{~~State~~} & Alice's & Bob's & Hardy's \\ 
& ~~Measurements~~ & ~~Measurements~~ & ~~Success~~\\ 
\hline\hline
\multirow{2}{1em}{$\Phi_{I}$} & $\mathcal{M}_1 = \{e_1,e_4\} $ & $\mathcal{N}_1 = \{e_6,e_3\}$  &  \multirow{2}{4em}{~~1/4} \\ 
& $\mathcal{M}_2 = \{e_2,e_5\} $ & $\mathcal{N}_2 = \{e_5,e_2\}$  & \\
\hline
\multirow{2}{1em}{$\Phi_{II}$} & $\mathcal{M}_1 = \{e_4,e_1\} $ & $\mathcal{N}_1 = \{e_4,e_1\}$  &  \multirow{2}{4em}{~~1/20} \\ 
& $\mathcal{M}_2 = \{e_2,e_5\} $ & $\mathcal{N}_2 = \{e_2,e_5\}$  &  \\
\hline
\multirow{4}{1em}{$\Phi_{III}$} & $\mathcal{M}_1 = \{e_3,e_6\} $ & $\mathcal{N}_1 = \{e_1,e_4\}$  &  \multirow{2}{4em}{~~1/16} \\ 
& $\mathcal{M}_2 = \{e_2,e_5\} $ & $\mathcal{N}_2 = \{e_6,e_3\}$  &  \\
\cline{2-4}
& $\mathcal{M}_1 = \{e_4,e_1\} $ & $\mathcal{N}_1 = \{e_1,e_4\}$  &  \multirow{2}{4em}{~~1/8} \\ 
& $\mathcal{M}_2 = \{e_2,e_5\} $ & $\mathcal{N}_2 = \{e_5,e_2\}$  &  \\
\hline
\multirow{4}{1em}{$\Phi_{IV}$} & $\mathcal{M}_1 = \{e_3,e_6\} $ & $\mathcal{N}_1 = \{e_4,e_1\}$  &  \multirow{2}{4em}{~~1/28} \\ 
& $\mathcal{M}_2 = \{e_2,e_5\} $ & $\mathcal{N}_2 = \{e_3,e_6\}$  &  \\
\cline{2-4}
& $\mathcal{M}_1 = \{e_4,e_1\} $ & $\mathcal{N}_1 = \{e_4,e_1\}$  &  \multirow{2}{4em}{~~1/14} \\ 
& $\mathcal{M}_2 = \{e_2,e_5\} $ & $\mathcal{N}_2 = \{e_2,e_5\}$  &  \\
\hline
\multirow{4}{1em}{$\Phi_{V}$} & $\mathcal{M}_1 = \{e_4,e_1\} $ & $\mathcal{N}_1 = \{e_3,e_6\}$  &  \multirow{2}{4em}{~~1/28} \\ 
& $\mathcal{M}_2 = \{e_3,e_5\} $ & $\mathcal{N}_2 = \{e_2,e_5\}$  &  \\
\cline{2-4}
& $\mathcal{M}_1 = \{e_4,e_1\} $ & $\mathcal{N}_1 = \{e_4,e_1\}$  &  \multirow{2}{4em}{~~1/14} \\ 
& $\mathcal{M}_2 = \{e_2,e_5\} $ & $\mathcal{N}_2 = \{e_2,e_5\}$  &  \\
\hline
\multirow{2}{1em}{$\Phi_{VI}$} & $\mathcal{M}_1 = \{e_2,e_5\} $ & $\mathcal{N}_1 = \{e_2,e_5\}$  & \multirow{2}{4em}{~~1/8} \\ 
& $\mathcal{M}_2 = \{e_1,e_4\} $ & $\mathcal{N}_2 = \{e_1,e_4\}$  &  \\
\hline
\end{tabular}
\caption{Measurement choices for Alice and Bob and the corresponding Hardy's success probabilities for the six different classes of entangled states in bipartite hexagon theory.}\label{tab1}
\end{table}

\subsection{Mixed entangled state and Hardy's nonlocality}
We will now consider the possibilities of the Hardy-type nonlocality for those preparation devices that produce \textcolor{blue}{a} pure entangled g-bit \footnote{Analogous to the terms c-bit and qubit for two-level classical and quantum systems respectively, here we adopt the phrase \textit{g-bit} to denote the elementary system of a GPT, with OD=2.} or a pure product g-bit with a pre-declared ignorance, and hence result in a mixed preparation. In quantum theory, such a preparation devise which prepares a two-qubit pure entangled state and a two-qubit pure product state with some predefined ignorance doesn't exhibit Hardy Nonlocality. (Technically any two-qubit mixed state does not exhibit Hardy-type nonlocality \cite{Kar97}. But here we are interested in these special preparation devices as they can help us study some topological properties of the underlying GPT). However, we will show this is not the case in polygonal GPT models. In the following, we will identify the \textit{preparation-measurement reciprocity} \cite{Ballentine14} as the salient feature of quantum theory, which causes such a difference. From now on, by the phrase "mixed entangled state" we would like to mean a convex combination of pure entangled state along with a product g-bit. For the sake of completeness, let us begin with a definition of preparation-measurement reciprocity in quantum theory.
\begin{definition}
[Preparation-measurement reciprocity \cite{Ballentine14}] For every pure quantum preparation  $\omega_{\psi}=\ketbra{\psi}{\psi}$ the corresponding effect, i.e.,  $e_{\psi}=\ketbra{\psi}{\psi}$ clicks certainly whenever the measurement $\mathcal{M}:=\{e_{\psi},u-e_{\psi}\}$ is performed. Furthermore, $\omega_{\psi}$ is the only preparation which passes the effect $e_{\psi}$ certainly.    
\end{definition}
The notion can readily be extended in GPTs. We say a GPT satisfies 'preparation-measurement reciprocity' if for every pure state $\omega$ there exists a unique extremal effect $e_\omega$ that filters the state $\omega$ with certainty. Also note that all such GPT models can be assigned with a finite dimensional linear vector space and hence the effects, i.e., the linear functional on this vector space allows a one-to-one correspondence with the state vectors, which is referred as \textit{weak self-duality} \cite{Janotta11, Barnum08}. This, in turn, further strengthens the condition of preparation-measurement reciprocity in such GPT models.
\begin{definition}
[Hardy-type local correlation] A correlation is called Hardy-type local if it is local and it satisfies Eqs. (\ref{h2})-(\ref{h4}).  
\end{definition}
\begin{definition}
[Trivial Hardy-local theory] A bipartite theory is said to be trivial Hardy-local if no product state results in a Hardy-type local correlation under incompatible measurements performed on the subsystems.   
\end{definition}
With these definitions, we are now in a position to prove an important result. 
\begin{lemma}\label{l1}
The bipartite composition of a GPT will be a 'trivial Hardy-local theory' whenever its local parts have $OD=2$ and satisfy preparation-measurement reciprocity.
\end{lemma}
\begin{proof}
Consider a bipartite GPT in which the local parts have $OD=2$ and satisfy preparation-measurement reciprocity. Let's assume there exists a pure product state $\omega_A\otimes\omega_B$ that can give rise to a  Hardy-type local correlation when incompatible measurements $\{\mathcal{M}_1^A,\mathcal{M}_2^A\}$ and $\{\mathcal{N}_1^B,\mathcal{N}_2^B\}$ are performed on the subsystems. We denote these measurements as $\{e_i^k,u-e_i^k\}$ with $i\in\{1,2\}$ and $k\in\{A,B\}$. Without loss of any generality, we can assume that the effects $\{e_i^k,u-e_i^k\}$ are extreme. The preparation-measurement reciprocity demands that for a pure state $\omega$ on the local part there exists a unique extremal effect $e_{\omega}$ which filters the state $\omega$ perfectly. Now, if the theories with local OD$=2$ are considered, the extremal complementary effect corresponding to  $e_{\omega}$ is given by $u-e_{\omega}$ which must be unique as well.\footnote{One can interpret the preparation-measurement reciprocity theorem in an alternative way for the theories with OD$=2$ : for every pure state $\omega$,  
$(u-e_{\omega})$ is the unique extreme effect that filters $\omega$ with zero probability. Note that, this does not hold good for OD$\geq3$. For instance, one can consider bipartite qutrit systems which is not a trivial Hardy-local theory, in spite of admitting preparation-measurement reciprocity.}

Let us now, write down the equations for $\omega_A\otimes\omega_B$ to create a Hardy-type local correlation.
\begin{subequations}
\begin{align}
\nonumber\Tr(e_1^A\omega_A)\times\Tr(e_1^B\omega_B)=0\\\nonumber
\Tr(e_1^A\omega_A)\times\Tr(e_2^B\omega_B)=0\\\nonumber%\label{h2},\\
\Tr(e_2^A\omega_A)\times\Tr(e_1^B\omega_B)=0
\\\nonumber%\label{h3},\\
\Tr[(u-e_2^A)\omega_A]\times\Tr[(u-e_2^B)\omega_B]=0.%\label{h4}.
\end{align}
\end{subequations}
Here the first equality follows from the fact that this correlation must be local. Let $\Tr(e_1^A\omega_A)=0$, then preparation-measurement reciprocity, together with the fact that the OD of the local systems is exactly $2$, implies $\Tr(e_2^A\omega_A),\Tr[(u-e_2^A)\omega_A]\neq0$. Thus we must have $\Tr(e_1^B\omega_B)=\Tr[(u-e_2^B)\omega_B]=0$. This is a contradiction as system $B$ also obeys preparation-measurement reciprocity. A similar contradiction will arise if we assume $\Tr(e_1^B\omega_B)=0$. This completes the proof.
\end{proof}
With the help of the above Lemma, we will finally conclude that,
\begin{theorem}\label{t3}
A mixed preparation device, for a theory (with OD=2 for each subsystem) obeying preparation-measurement reciprocity on each part, does not exhibit Hardy nonlocality.
\end{theorem}
\begin{proof}
A local correlation can not exhibit Hardy nonlocality by construction. However, there are deterministic local correlations that satisfy the zero conditions for Hardy nonlocality, {\it i.e.}, Eq. (\ref{h2}) - (\ref{h4}).

If our mixed state always consists of a pure product state for a preparation-measurement reciprocal theory, it follows from Lemma \ref{l1} that either $\mathcal{M}_1^A$ and $\mathcal{M}_2^A$ or $\mathcal{N}_1^B$ and $\mathcal{N}_2^B$ are the compatible pair of measurements.%\st{the only possible measurement choice is either, $\mathcal{M}_1^A\simeq \mathcal{M}_2^A$ or, $\mathcal{N}_1^B\simeq\mathcal{N}_2^B$.}} 
Under such a condition no entangled state can exhibit Hardy nonlocality, otherwise, it leads to superluminal communication \cite{Andersson05,Busch86,Kar16}. This completes the proof.
\end{proof}

%\footnote{\textcolor{blue}{It is important to note that, both Lemma \ref{l1} and Theorem \ref{t3} hold only for those preparation-measurement reciprocal theories, for which the \textit{operational dimension} corresponding to the local systems is exactly $2$. For instance, one can consider the three dimensional quantum systems, the bipartite composition of which is not a trivial Hardy-local theory, in spite of admitting preparation-measurement reciprocity.}}

Now, we will show that bipartite compositions of all the discrete operational models with an even number of pure states and specifically the pentagon model exhibit Hardy-type nonlocality when a suitably chosen pure product preparation probabilistically sampled with a pure entangled preparation. Since, for all such theories the OD is exactly $2$, we can identify the absence of preparation-measurement reciprocity in their topology. We conjecture the same feature also holds for any odd $n$-gon models, however, due to the numerical limitations we are bound with $n=5$ case only.   
\begin{theorem}\label{theo3}
For every even ($n$)-gon theory and  $\forall\epsilon\in(0,1]$, there exists a class of mixed entangled states $W_{\epsilon}=\epsilon \Phi_{J}+ (1-\epsilon) \omega_{i}\otimes\omega_{j}$ exhibiting Hardy nonlocality.
\end{theorem}  
\begin{proof}
From Theorem \ref{theo2} it follows that to obtain a Hardy nonlocal behaviour from the state $\Phi_{J}$ it requires $E_{1}^{A+}=e_{r}$ then $E_{2}^{A+}=e_{r\pm1}$. In an even-gon theory every two consecutive effects $e_{i}$ and $e_{i\oplus_{n}1}$ click with certainty on the state $\omega_{i}$ and they never click on state $\omega_{i\oplus_{n}n}$ (here modulo operation is same as defined in footnote of page 7).
Therefore, $\forall r\in\{1,2,\cdots,n\}$ there is a state $\omega_{r}\in\Omega^{A}$ in Alice's side, such that $\Tr((E_{1}^{A+})^{T}\omega_{r})=\Tr((E_{2}^{A+})^{T}\omega_{r})=0$. Similarly, there is a state $\omega_{s}\in\Omega^{B}$, such that $\Tr((E_{2}^{B-})^{T}\omega_{s})=0$. Evidently, the state $\omega_{r}\otimes\omega_{s}$ satisfies all the conditions Eq. (\ref{ho2})-(\ref{ho4}) and equals to zero for Eq. (\ref{ho1}). So, for the state  $W_{\epsilon}=\epsilon\Phi_{J}+ (1-\epsilon)\omega_{r}\otimes\omega_{s}$, with $\epsilon\in(0,1]$, all the conditions (\ref{ho2}) - (\ref{ho4}) are satisfied, with  $\tr[(E_{1}^{A+}\otimes E_{1}^{B+})^\mathrm{T}W_{\epsilon}]=\epsilon\times \sin^{2}(\frac{\pi}{n})$. Clearly, for $\epsilon\neq1$, the state $W_{\epsilon}$ is a mixed state and hence establishes the claim of the theorem.  
\end{proof}
In a similar spirit, it is possible to show that the mixed state of odd-gon theories can also exhibit Hardy's nonlocality. In the following, we give a proof for the bipartite pentagon theory.
\begin{theorem}\label{theo4}
For every value of $\epsilon\in(0,1]$, the mixed entangled state $W_{\epsilon}=\epsilon \Phi_{H}+ (1-\epsilon)\omega_{i}\otimes\omega_{j}$ exhibits Hardy-type nonlocality for suitable choice of measurement, whenever $\omega_{i}\otimes\omega_{j} \in \{\omega_{3}\otimes\omega_{4},\omega_{3}\otimes\omega_{5},\omega_{4}\otimes\omega_{3},\omega_{4}\otimes\omega_{4},\omega_{5}\otimes\omega_{3}\}$.
\end{theorem}
\begin{proof}
If we consider two incompatible measurements $\mathcal{M}_1\equiv\{e_1,\bar{e}_1\}$ and $\mathcal{M}_2\equiv\{\bar{e}_2,e_2\}$ on Alice's part and two incompatible measurements $\mathcal{N}_1\equiv\{e_1,\bar{e}_1\}$ and $\mathcal{N}_2\equiv\{\bar{e}_2,e_2\}$ on Bob's part, then the resulting correlation obtained from the state $W_\epsilon$ depicts Hardy's nonlocality, whenever $\omega_{i}\otimes\omega_{j}\in \{\omega_{3}\otimes\omega_{4},\omega_{4}\otimes\omega_{3},\omega_{4}\otimes\omega_{4}\}$. For $\omega_{i}\otimes\omega_{j}=\omega_{3}\otimes\omega_{5}$, we require the measurements $\mathcal{M}_1\equiv\{e_1,\bar{e}_1\}$, $\mathcal{M}_2\equiv\{e_5,\bar{e}_5\}$ on Alice's part and $\mathcal{N}_1\equiv\{e_1,\bar{e}_1\}$, $\mathcal{N}_2\equiv\{\bar{e}_2,e_2\}$ on Bob's part; and for $\omega_{i}\otimes\omega_{j}=\omega_{5}\otimes\omega_{3}$, $\mathcal{M}_1\equiv\{e_1,\bar{e}_1\}$, $\mathcal{M}_2\equiv\{\bar{e}_2,e_2\}$ on Alice's side and $\mathcal{N}_1\equiv\{e_1,\bar{e}_1\}$, $\mathcal{N}_2\equiv\{e_5,\bar{e}_5\}$ on Bob's side suffice the purpose. The success probability turns out to be $P^{succ}=\epsilon\times(1-\frac{4\sqrt{5}}{10})\approx0.1056\epsilon$.
\end{proof}
While from the perspective of Bell-nonlocality, the polygonal state spaces exhibit no characteristic distinction (except the quantitative bounds) from their continuous counterpart QT, Theorem \ref{theo3} and \ref{theo4} exhibit such a distinction for Hardy-type nonlocal arguments. However, the signature of such a difference vanishes considering the bipartite compositions of higher quantum systems. In particular, for higher dimensional QT, there are incompatible local measurements with the proper choice of separable bipartite state, which can generate any of the extreme local correlations, and hence there are mixed entangled states depicting Hardy-type nonlocal arguments. This, in turn, directs towards the state space topology of the qubit system and the continuity therein to demonstrate it as unique among the possible two-dimensional state space structures.

\section{Hardy nonlocality and almost quantum correlation}\label{five}
We will now delve into characterizing a set of almost quantum correlations obtained from the bipartite maximal compositions of the odd-gon theories, specifically that for which the local state-space can be described with exactly \textit{five} pure preparations. Mathematically, a non-signalling behaviour is said to be an almost quantum correlation, if and only if it can be realized by a couple of projectors acting on a normalized vector of a joint Hilbert space and the statistics obtained from the projector-vector pair is independent of party-permutation \cite{Navascues15}. This readily confirms that quantum correlations are the strict subset of such correlations, where the second condition is satisfied by the structures of commutative algebra. Note that the general characterization of the almost quantum correlations is possible via a semidefinite programming (see the \textit{methods} section in \cite{Navascues15}).
\begin{definition}\label{defi2}
A bipartite correlation p(ab|xy) is said to be an almost quantum correlation if there exists a Hilbert space $\mathcal{H}$, a normalized state $\ket{\phi}\in \mathcal{H}$ and projection operators ${E^a_x},{F^b_y} \in B(\mathcal{H})$ such that the following conditions hold
\begin{itemize}
    \item $\sum_a E^a_{x} = \sum_b F^b_{y} = \mathbb{I}_\mathcal{H},~  \forall~x,y$
    \item $E^a_{x}F^b_{y}\ket{\phi}=F^b_{y}E^a_{x}\ket{\phi},~\forall~a,b,x,y$
    \item $p(ab|xy)=\bra{\phi}E^a_{x}F^b_{y}\ket{\phi},~\forall~a,b,x,y$.
\end{itemize}
\end{definition}
Despite having a compact mathematical description there is no known operational model, prescribed for almost quantum correlations. Moreover, an operational model, which contains all the mathematical consistent states and effects is unable to reproduce the complete set of almost quantum correlations \cite{Sainz18}.

However, in the following, we will show that the strict subset of $2$-input-$2$-output almost quantum correlations can be obtained from the non-maximally entangled state of the bipartite composition of local pentagon models.
\begin{proposition}
A continuous family of almost quantum correlations can be obtained from the maximal bipartite composition of the pentagon model.
\end{proposition}
\begin{proof}
We begin with a numerical observation: the SDP formulation for characterizing an almost quantum correlation confirms that any $2$-input-$2$-output non-signaling correlations obtained from the non-maximally pure entangled state $\Phi_{H}$ (see Eq. (\ref{pentaH})) resides inside the set of almost quantum correlations. On the other hand, the correlations obtained from any product state is a local one, and the probabilistic mixture of which with any almost quantum correlation is also an almost quantum correlation point. This follows from the fact that almost quantum correlations are closed under classical post-processing and wirings \cite{Navascues15}.

Therefore, the convex combination of the pure product state $\omega_5\otimes\omega_3$ and $\Phi_H$ under the local measurement settings $\mathcal{M}_1\equiv\{e_1,\bar{e}_1\}$ and $\mathcal{M}_2\equiv\{e_2,\bar{e}_2\}$ for Alice and $\mathcal{N}_1\equiv\{e_1,\bar{e}_1\}$ and $\mathcal{N}_2\equiv\{\bar{e}_5,e_5\}$ for Bob will produce a set of almost quantum correlations. Also, it follows from Theorem \ref{theo4} that such a correlation will exhibit Hardy-type nonlocality with a violation of $0.1056\epsilon$, where $\epsilon$ is the probability of the state $\Phi_H$ in the convex mixture. Since the quantum theory does not allow Hardy violation to be larger than $0.0902$, the above-mentioned convex combination will be strictly almost quantum (almost quantum but not quantum) when $\epsilon>\frac{0.0902}{0.1056}=0.8542$. This implies that for $\epsilon\in(0.8542,1]$ the correlations obtained from the state $\epsilon \Phi_H + (1-\epsilon) \omega_5\otimes\omega_3$, in the bipartite composition of local pentagon model, is a subset of almost quantum (but not quantum) correlations.
\end{proof}
Note that, the $2-2-2$ correlations obtained from the same state and measurements, for the region $\epsilon\in(0,0.8542]$, can not be trivially concluded as quantum correlations. Answering such a question can be a potential direction to understand the possible overlap between these operational models and quantum theory in their (nonlocal) correlation space. This requires further analysis in terms of properly devised Bell-type inequalities, which is beyond the scope of the present paper. However, in the following we will confirm about their nonlocal overlap, considering another form set of $2-2-2$ correlations.

\section{Intersection of odd-gon with quantum}\label{six}
Here we show that there exist correlations which are present in both odd-gon theory and Quantum theory. In fact, we can prove that given any nonlocal theory, there exist correlations which are present in that theory as well as quantum theory. Needless to say, that such an overlap will be trivially there for the set of local correlations. However, the intersection for the nonlocal region can be proved with a very simple twirling argument. Given any nonlocal correlation $p(ab|xy)$. We can consider the following set of Local relabelings of $p(ab|xy)$. 
\begin{subequations}
\begin{align}
\nonumber
R_0&\equiv a'=a\oplus1,b'=b\oplus1\\
\nonumber
R_1&\equiv x'=x\oplus1,b'=b\oplus y\\
\nonumber
R_2&\equiv y'=y\oplus1,a'=a\oplus x\\
\nonumber
R_3&\equiv\begin{cases} 
x'=x\oplus1,y'=y\oplus1,\\
a'=a\oplus x,b'=b\oplus y\oplus1
\end{cases}
\end{align}
\end{subequations}
It can be argued that under any of the local relabelings $\{R_0,\cdots R_3\}$ the $B_{CHSH}$ value of the correlation does not change. Now we construct a new correlation $q_{iso}(ab|xy)$ given by
\begin{align}
q_{iso}(ab|xy)=\frac{1}{8}[(1+R_1+R_2+R_3)p(ab|xy)\nonumber\\
        ~~~~+R_0(1+R_1+R_2+R_3)p(ab|xy)]
\end{align}
It can be seen that such a correlation will have the same $B_{CHSH}$ value as $p(ab|xy)$ and it will also be an isotropic correlation. Thus we have established that given any nonlocal correlation $p(ab|xy)$ in any theory, the theory also exhibits the isotropic correlation $q_{iso}(ab|xy)$ by mixing $8$ different local relabelings of $p(ab|xy)$. Moreover, the nonlocal content under such an operation remains invariant. Thus any nonlocal theory must have some overlap with the isotropic quantum correlations. This completes the proof.

\section{Inequivalence of entanglement and B-CHSH nonlocality in polygon models}\label{seven}
While all bipartite quantum pure states exhibit nonlocality \cite{Gisin91}, the pure states are too idealistic when experimental situations are considered. So naturally the question arises whether mixed states exhibit such nonlocal behaviour. A particular family that is of interest to us is the Werner class of states
\begin{equation}
\mathcal{W}_{p}=p\ket{\psi^{-}}_{AB}\bra{\psi^{-}} + (1-p)\frac{\mathcal{I}}{2}\otimes \frac{\mathcal{I}}{2},
\end{equation}
where $\ket{\psi^-}:=(\ket{01}-\ket{10})/\sqrt{2}\in\mathbb{C}^2\otimes\mathbb{C}^2$ and $p\in[-1/3,1]$. In particular, for $p\in[0,1]$ the state can be thought of as a statistical mixture of the singlet state and white noise. Straightforward calculation yields the state $\mathcal{W}_{p}$ is entangled for $p>\frac{1}{3}$ and violates CHSH inequality for $p>\frac{1}{\sqrt{2}}$. In a seminal result, Werner established that for $\frac{1}{3}< p \leq \frac{1}{2}$ the statistics obtained from the state $\mathcal{W}_{p}$ through local projective measurements can be explained by local hidden variable model \cite{Werner89}. Later Barrett extended this model for arbitrary local measurement for the parameter range $p \leq \frac{5}{12}$ \cite{Barret02} (see also \cite{Rai12}). This result is quite important as it establishes that entanglement and nonlocality as two inequivalent notions. A similar question one can ask in polygon theories. Our next result partially addresses this question.
\begin{theorem}
For all the theories where $n>4$, there exists a class of mixed entangled states that does not violate CHSH inequality.
\end{theorem}
\begin{proof}
Consider the class of states, $\mathcal{W}^{O}_{p}:=p\Phi_{J}+(1-p)u\otimes u$,  where $p\in[0,1]$ and $\Phi_J$ is the state given in Eq. (\ref{oddJ}). Clearly, $\mathcal{W}^{O}_{p}$ is a mixed state whenever $p\in[0,1)$. 

For the odd-gons, the expectation value of any measurement $\langle\mathcal{M}_i\mathcal{N}_j\rangle$ on the maximally mixed state $u\otimes u$ reads as $\left(\frac{r_n^2-1}{r_n^2+1}\right)^2$.
For the state $\mathcal{W}^{O}_{p}$ the maximum value of Bell-CHSH expression becomes $\mathbb{B}_{\max}(\mathcal{W}^{O}_{p})=p\mathbb{B}^{(n)}_{\max}+2(1-p)\left(\frac{r_n^2-1}{r_n^2+1}\right)^2$, where $\mathbb{B}^{(n)}_{\max}$ is the maximum Bell-CHSH value obtained from the $n$-gonal maximum entangled state $\Phi_J$. Denoting the range of the parameter $p_{NL}$ of the state $\mathcal{W}^{O}_{p}$ showing Bell-CHSH nonlocality we have 
\begin{align}\label{eq20}
p_{NL}  > 8r_n^2\left[\mathbb{B}^{(n)}_{\max}\left(r_n^2+1\right)^2-2\left(r_n^2-1\right)^2\right]^{-1}.
\end{align}
We now proceed to find the range of the parameter $p$ of the state $\mathcal{W}_{p}^O$ for which the state is entangled. Note that, unlike quantum theory, in this model, we do not have any criterion like negative partial transposition (NPT) \cite{Peres96,Horodecki96} that can detect the entanglement of a state. However, if we can find an effect that is entangled and yields a negative probability on some state, then by definition, the state must be entangled\footnote{At this point, an observant reader should note that in quantum theory all the entangled states yield non-negative probability on all the entangled effects. This is due to the fact that state and effect cones are self-dual. However, in abstract GPT, this might not be the case, which in turn results in different consistent compositions for the same elementary systems. At this point the Refs.\cite{DallArno17,Saha20,Naik22,Lobo21,Sen2022,Patra2023} are worth mentioning.}. This is because any product state on an entangled effect always gives a non-negative probability. For the odd-gon theory, it has been shown that the effects $E_{ab}$ and $\bar{E}_{ab}:=u\otimes u-E_{ab}$ are entangled \cite{Saha20(1)}, where  
\begin{equation*}
E_{ab}=\frac{1}{1+r_{n}^2}
\begin{pmatrix}
1 & 0 & 0\\
0 & 1 & 0\\
0 & 0 & 1
\end{pmatrix}  . 
 \end{equation*}
A straightforward calculation yields  $\Tr[\bar{E}_{ab}^T\mathcal{W}_{p}^O]=p\left(\frac{r_n^2-2}{1+r_n^2}\right)+(1-p)\frac{r_n^2}{1+r_n^2}$. Denoting the range of parameter of the state $\mathcal{W}_{p}^O$ as $p_E$ for which the state must be entangled we have 
\begin{align}\label{eq21}
\Tr[\bar{E}_{ab}^T\mathcal{W}_{p_E}^O]<0\implies p_E >\frac{r_n^2}{2}.
 \end{align}
Comparing Eq. (\ref{eq20}) and Eq. (\ref{eq21}), it is evident that $p_E<p_{NL}~\forall~\mbox{odd}~n$. Therefore within the range $p_E\leq p<p_{NL}$ the state $\mathcal{W}_{p_E}^O$ is entangled but it does not violate Bell-CHSH inequality.

For the even gon theory we consider the state $\mathcal{W}^{E}_{p}:=p\Phi_{J}+(1-p)u\otimes u$,  where $p\in[0,1]$ and $\Phi_J$ is the state given in Eq. (\ref{evenJ}). Noting the fact that $\langle\mathcal{M}_i\mathcal{N}_j\rangle$ on $u\otimes u$ turns out to be zero in this case and using the entangled effect 
\begin{equation*}
E_{ab}=\frac{1}{2}
\begin{pmatrix}
-\cos\frac{\pi}{n} & -\sin\frac{\pi}{n} & 0\\
\sin\frac{\pi}{n} & -\cos\frac{\pi}{n} & 0\\
0 & 0 & 1
\end{pmatrix}   
\end{equation*}
identified in \cite{Saha20(1)}, a similar calculation yields 
\begin{equation}\label{eq22}
p_{NL}>\frac{2}{\mathbb{B}^{(n)}_{\max}},~~\&~~p_E > \frac{1}{2}, 
\end{equation}
where $\mathbb{B}_{\max}^{(n)}$ is the maximal Bell-CHSH value for the state $\Phi_J$. Since for all the even gons $\mathbb{B}_{\max}^{(n)}<4$, whenever $n>4$ \cite{Janotta11}, therefore within the range $p_E\leq p<p_{NL}$ the state $\mathcal{W}^{E}_{p}$ does not violate Bell-CHSH inequality although it is entangled. This completes the proof.  
\end{proof}
Important to note that, for $n=4$, we Have  $\mathbb{B}_{\max}^{(4)}=4$, which thus leads to the fact that all the entangled states in this theory are Bell-CHSH nonlocal. 
 
\section{Discussions}\label{eight}
Our investigation centered on the distinct classes of extreme bipartite states within locally regular polygon systems, examining their nonlocal behavior. Employing the methodology outlined in Section (\ref{IIIA}), we demonstrated that the bipartite square state space comprises a sole class of entangled states. Notably, for local pentagons and hexagons, we identified two and six categories of entangled states, respectively. Furthermore, we observed that while states represented by Eq. (\ref{oddJ}) and Eq. (\ref{evenJ}), alongside their equivalent counterparts, share similarities with maximally entangled two-qubit states, the newly identified classes resemble non-maximally entangled states of two qubits.

We delved into exploring the nonlocality of polygon theories through Hardy's nonlocal argument. Surprisingly, in contrast to bipartite quantum theory with local OD$=2$, we found that maximally entangled states associated with even-gon theories satisfy the mentioned argument with a success probability of $\sin^2{(\frac{\pi}{n})}$. As the number of vertices in even-gon compositions increases, the success probability decreases accordingly. Conversely, we did not observe a similar nonlocal scenario in the case of bipartite odd-gon systems with maximally entangled states. However, within the bipartite pentagon theory, we discovered a higher degree of nonlocality with the non-maximally entangled states $\Phi_H$ compared to the maximum success probability attainable through quantum theory. This finding places the correlation outside the realm of quantum phenomena, although it remains within the domain of almost quantum correlations $\Tilde{Q}$.

Our analysis revealed that akin to quantum theory, nonlocality and entanglement remain distinct concepts in these discrete theories, especially when the number of vertices in the local state spaces exceeds four. This conclusion was established by utilizing Werner-like mixed states, probabilistic mixtures of maximally entangled states with the maximally mixed ones. Notably, while these theories showcase a disparity between entanglement and nonlocality in the probability range of the entangled state, this disparity vanishes in the case of the composite square theory.

Moreover, beyond the hexagon state space, our research opens up several questions for further investigation. While we propose a method for examining the existence of various classes of extremal entangled states, the approach becomes computationally inefficient for arbitrary higher-gon theories. Developing a more efficient method applicable to these theories is imperative. Additionally, the characterization of Hardy nonlocality in higher odd-gon theories remains an open question due to the incompleteness of our findings.

Furthermore, it would be intriguing to explore the potential existence of overlapping regions between the correlations obtained from bipartite local pentagon theories and those arising from bipartite qubit states. Specifically, understanding whether all correlations reproducible by bipartite qubit states can be considered as the common intersection of the local odd-gon theories or if there are overlapping regions among different theories would be of significant interest.

{\bf ACKNOWLEDGMENTS: }We would like to sincerely acknowledge the pertinent comments of the anonymous referee, helping us to formalize Theorem 3 of the present manuscript in a more comprehensive way. MK, TM, SGN, and MB thankfully acknowledge fruitful discussions with Mir Alimuddin, Edwin Peter Lobo, Ram Krishna Patra, and Samrat Sen. TG and SS would like to acknowledge fruitful discussions with Subhendu B. Ghosh. TG is supported by the Hong Kong Research Grant Council through Grants No. 17307719 and 17307520 and though the Senior Research Fellowship Scheme SRFS2021- 7S02, by the Croucher Foundation, and by the John Templeton Foundation through Grant No. 62312, The Quantum Information Structure of Space-time (qiss.fr). MB acknowledges support through the National Mission in Interdisciplinary Cyber-Physical Systems from the Department of Science and Technology through the IHUB Quantum Technology Foundation (Grant no: IHUB/PDF/2021-22/008).

\onecolumn\newpage
\appendix

\section{Finding the extreme entangled states}
As already discussed, the amount of computations needed to find all the extreme states grows drastically with the number of extreme states of the elementary polygon systems. However, we can cut down this search space significantly by looking into the structure of the problem. Firstly we note that finding all extreme states is not necessary to find the entanglement classes. Once we have one state from each class, we can find all other extreme states by the action of the Local reversible transformations. Here we discuss a method to directly find a representative element for different entanglement classes.

As noted earlier atleast 8 hyperplanes are required to find out an extreme state. These hyperplane equations essentially represent positivity conditions. For the bipartite $n$-gon system let us denote the set of all extreme product effects as $\mathcal{P}_{ef}[n]\equiv\{E_1,\dots,E_{n^2}\}$. Let the set $S_1\equiv \{a_i\}_{i=1}^8\subset \mathcal{P}_{ef}[n]$ lead to a solution state (extreme) $\omega_1$ and the set $S_2\equiv \{b_j\}_{j=1}^8\subset \mathcal{P}_{ef}[n]$ lead to a solution state $\omega_2$. 
\begin{definition}
[Local reversible equivalence] Two sets $S_1$ and $S_2$ will be called equivalent under local reversible transformation (LRT) if $\exists~s_1,s_2\in\{+,-\}$ and $k_1,k_2\in\{1,\dots,n\}$ such that $S_2=\bigcup_{i=1}^8\{(\mathcal{T}_{k_1}^{s_1}\otimes\mathcal{T}_{k_2}^{s_2})a_i~|~a_i\in S_1\}$.
\end{definition}
In such a case the corresponding solutions $\omega_1$ and $\omega_2$ must also be $LR$ equivalent, {\it i.e.}, $\omega_2=(\mathcal{T}_{k_1}^{s_1}\otimes\mathcal{T}_{k_2}^{s_2})\omega_1$. To formally characterise this $LR$ connected solutions in a systematic way we recall some preliminary concepts from group theory in the following subsection. 

\subsubsection{Preliminaries}
Let $\mathcal{G}$ be a finite group and let $O^{\mathcal{G}}$ denote a set of objects on which the group elements act. The set $O^{\mathcal{G}}$ is closed under action of group elements, {\it i.e.} $g(o_1)\in O^{\mathcal{G}},\forall~g\in\mathcal{G}~\&~o_1\in O^{\mathcal{G}}$.
\begin{definition}
[Fixed point] An object $f\in O^{\mathcal{G}}$ is the fixed point of $g\in \mathcal{G}$ if it remains unchanged under the action of $g$, {\it i.e.} $g(f) = f$.
\end{definition}
\begin{definition}
[Orbit] The orbit $\mathcal{O}_o$ of an object $o\in O^{\mathcal{G}}$ is given by the set
\begin{align*}
\mathcal{O}_o:=\bigcup_{g\in \mathcal{G}} \{g(o)\}.
\end{align*}
\end{definition}
It is straightforward to see that if two objects $o_1$ and $o_2$ are related by some group action then $\mathcal{O}_{o_1}=\mathcal{O}_{o_2}$. Thus the set of all orbits partitions the collection of objects $O$ into disjoint sets. Also, it can be noted that every object belongs to exactly one orbit. Here we recall the orbit-counting theorem by Burnside \cite{Burnside97}. 
\begin{lemma}
[Burnside's Lemma] For a group $\mathcal{G}$ acting on a collection of objects $O^{\mathcal{G}}$, the number of orbits is given by
\begin{align*}
    \mbox{\# of orbits} = \frac{1}{|\mathcal{G}|}\sum_{g\in\mathcal{G}} \mbox{\# of fixed points of g}
\end{align*}
where $|\mathcal{G}|$ is the cardinality of the group $\mathcal{G}$.
\end{lemma}

\subsubsection{Group structure of LRT in bipartite polygon models}
Since $\mathcal{T}_{k_1}^{s_1}$ represents rotations and reflections, {\it i.e.}, an element of dihedral group $\mathcal{D}_{2n}$, we can observe that $\left\{\mathcal{T}_{k_1}^{s_1}\otimes\mathcal{T}_{k_2}^{s_2}|s_1,s_2\in\{+,-\}~\mbox{and}~k_1,k_2\in\{1,\dots,n\}\right\}$ also forms a group which we denote by $\mathcal{D}_{2n}^{\times2}$. We use this notation because this group is isomorphic to the group formed by the cartesian product of $\mathcal{D}_{2n}$ with itself. That is $\mathcal{D}_{2n}^{\times2}\ \cong \mathcal{D}_{2n} \times \mathcal{D}_{2n}$.
Now we define the set of objects $O^{\mathcal{D}_{2n}^{\times2}}$ as 
\begin{align*}
O^{\mathcal{D}_{2n}^{\times2}}=\left\{\left\{a_1,\dots,a_8\right\}~|~a_i\in\mathcal{P}_{ef}[n]\right\}.
\end{align*}
%\textcolor{red}{[$E_{Pro}$ has to be defined(???)]}

An object $\{f_1,\dots,f_8\}$ is called a fixed point of $\mathcal{T}_{k_1}^{s_1}\otimes\mathcal{T}_{k_2}^{s_2}\in\mathcal{D}_{2n}^{\times2}$ if
\begin{align*}
\mathcal{T}_{k_1}^{s_1}\otimes\mathcal{T}_{k_2}^{s_2}\left(\left\{f_1,\dots,f_8\right\}\right)=\left\{f_1,\dots,f_8\right\}.
\end{align*}
The orbit of an object $\left\{h_1,\dots,h_8\right\}\in O^{\mathcal{D}_{2n}^{\times2}}$ is given by 
\begin{align*}
    \mathcal{O}_{\left\{h_1,\dots,h_8\right\}}=\bigcup_{g\in \mathcal{D}_{2n}^{\times2}} \left\{g\left(\left\{h_1,\dots,h_8\right\}\right)\right\}.
\end{align*}
Since objects of one orbit are not connected to objects of another orbit by LRT, therefore, instead of evaluating all the $^{n^2}C_8$ possible cases for finding valid solutions, we can restrict ourselves to evaluating just one object from each orbit. This helps us in reducing the search space drastically. A list comparing this reduction is shown in Table \ref{tab2}. 
\begin{table}[t]
\centering
\begin{tabular}{c|c|c}
\hline
~~~~~n~~~~~ & $^{n^2}C_8$ & $\#$ of orbits  \\
\hline
$4$ & $12870$ & $283$\\
\hline
$5$ & $1081575$ & $11103$\\
\hline
$6$ & ~~~~~$30260340$~~~~~ & ~~~~~$213962$~~~~~ \\
\hline
\end{tabular}
\caption{Number of orbits for bipartite polygon systems.}\label{tab2}
\end{table}
Let $\mathbb{O}$ denote set of all orbits and let $x(\mathcal{O})$ be the representative object of an orbit $\mathcal{O}\in\mathbb{O}$. We can now define a set  $\mathcal{X}_1$ as
\begin{align*}
    \mathcal{X}_1 = \{x(\mathcal{O})~|~\mathcal{O}\in\mathbb{O}\}.  
\end{align*}
Now we check whether the elements in $\mathcal{X}_1$ lead to a unique solution, {\it i.e.} whether the 8 hyperplanes corresponding to this object intersect at exactly one point. Using this we define the set $\mathcal{X}_2$ as
\begin{align*}
    \mathcal{X}_2 = \left\{y_1\in\mathcal{X}_1 ~|~\mbox{$y_1$ leads to a unique solution}\right\}  .
\end{align*}
The intersection of $8$ planes doesn't necessarily imply that the solution satisfies all the positivity conditions since the intersection of the 8 hyper-planes could lie outside the state space. So we define $\mathcal{X}_3$ as
\begin{align*}
    \mathcal{X}_3 = \left\{y_2\in \mathcal{X}_2 ~|~\mbox{$y_2$ satisfies all     positivity conditions}\right\} . 
\end{align*}
The elements in $\mathcal{X}_3$ are then analyzed to check if any of them are connected by local reversible transformations, which then lead to the end result of entangled states representative of each entanglement class.

\subsubsection{Orbit counting: Box world}
In the elementary Box world theory reversible transformations are the four rotations about the perpendicular axis passing through the centre, {\it i.e.} rotation about $0,~\pi/2,~\pi$, and $3\pi/2$ radians; and four reflections (along the two diagonals and the two lines connecting the midpoints of the parallel sides). These four reflections can also be obtained by fixing only one reflections and then followed by four rotations. We denote the four different rotations by $\left\{\mathbb{I},r,r^2,r^3\right\}$. Let us take the reflection that takes the effects $e_0,e_1,e_2$ and $e_3$ to  $e_3,e_2,e_1$ and $e_0$, respectively to be $f$\footnote{Please note that here the sub-index of $e_i$ takes values from $\{0,\cdots,3\}$, whereas in main manuscript it takes values from $\{1,\cdots,4\}$. However, this does not change the orbit counting.}. Then the four reflections are given by $\left\{f,rf,r^2f,r^3f\right\}$, yielding the full set of reversible transformations
\begin{align*}
\mathcal{R}=\left\{\mathbb{I},r,r^2,r^3,f,rf,r^2f,r^3f\right\},
\end{align*}
where $ab$ operation implies operation $b$ is followed by operation $a$. In the case of the composition of two such systems shared between Alice and Bob, we will denote the sets as 
\begin{align*}
\mathcal{R_A}&=\left\{\mathbb{I}_A,r_A,r^2_A,r^3_A,f_A,r_A f_A,r^2_A f_A,r^3_A f_A\right\},\\
\mathcal{R_B}&=\left\{\mathbb{I}_B,r_B~,r^2_B,r^3_B,f_B~,r_B f_B~,r^2_B f_B~,r^3_B f_B\right\}.
\end{align*}
Thus the set of all Local Reversible transformations on the composite system is given by
\begin{align*}
\mathcal{LR}_{box}=\left\{t_A\otimes t_B|t_A\in\mathcal{R_A},~t_B\in\mathcal{R_B}\right\}.
\end{align*}
Any product effect can be assigned a natural number using the rule $e_i\otimes e_j\to4i+j+1$, where $i,j\in\{0,1,2,3\}$, which is shown in Table \ref{tab3}.
\begin{table}[t!]
\centering
\begin{tabular}{c||c|c|c|c|}
~~~$\otimes$~~~&~~~$e^B_0$~~~&~~~$e^B_1$~~~&~~~$e^B_2$~~~&~~~$e^B_3$~~~ \\
\hline\hline
~~~$e^A_0$~~~&~~~$1$~~~&~~~$2$~~~&~~~$3$~~~&~~~$4$~~~ \\
\hline
~~~$e^A_1$~~~&~~~$5$~~~&~~~$6$~~~&~~~$7$~~~&~~~$8$~~~ \\
\hline
~~~$e^A_2$~~~&~~~$9$~~~&~~~$10$~~~&~~~$11$~~~&~~~$12$~~~ \\
\hline
~~~$e^A_3$~~~&~~~$13$~~~&~~~$14$~~~&~~~$15$~~~&~~~$16$~~~ \\
\hline
\end{tabular}
\caption{The effect $e_i\otimes e_j$ is assigned a natural number following the rule $e_i\otimes e_j\to4i+j+1$, where $i,j\in\{0,1,2,3\}$. For instance, $e^A_2\otimes e^B_3$ is assigned $12$ (third row fourth column).}
\label{tab3}
\end{table}
\begin{figure}[t!]
\centering
\includegraphics[width=0.65\textwidth]{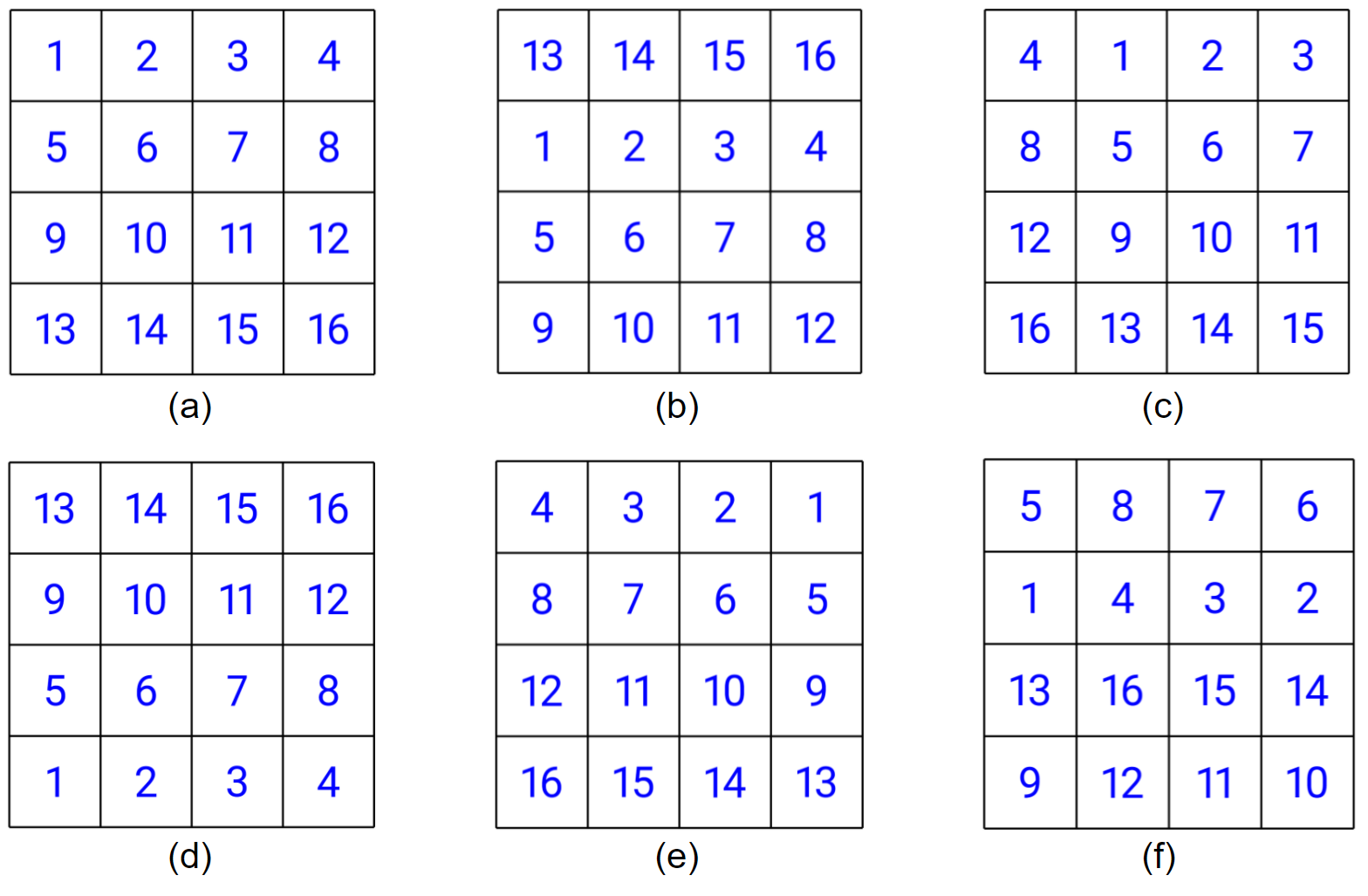}
\caption{Under the local reversible transformations $(\mathcal{LR}_{box})$, the arrangement of the numbers (associated with the product effect) in Table \ref{tab3} gets modified. Here we show few examples: (a) $\mathbb{I}_A\otimes\mathbb{I}_B$, (b) $ r_A\otimes\mathbb{I}_B$, (c) $\mathbb{I}_A\otimes r_B$, (d) $f_A\otimes\mathbb{I}_B$, (e) $\mathbb{I}_A\otimes f_B$, (f) $f_Ar_A^2\otimes f_Br_B^3$.}\label{box_7}
\end{figure}
With this notation, if Alice applies the transformation $r_A$ on her part then each row in Table \ref{tab3} steps downward and the last row wraps back to the first row. Similarly, for $r_B$ each column shifts one step rightwards and the last column wraps back to the first column. On the other hand, the operations $f_A$ reflects the Table \ref{tab3}  about the central horizontal line, {\it i.e.} $\mbox{row}-1\leftrightarrow \mbox{row}-4$ and $\mbox{row}-2\leftrightarrow \mbox{row}-3$; and under $f_B$ we have $\mbox{column}-1\leftrightarrow \mbox{column}-4$ and $\mbox{column}-2\leftrightarrow \mbox{column}-3$. All other local reversible transformations of Table \ref{tab3} can be obtained by suitable combinations of these elementary operations (see Fig.\ref{box_7}).

We now move on to calculating the number of orbits. For that, according to Burnside's Lemma, we need to find the number of fixed points for every local reversible transformation. Consider the transformation $r_Af_A\otimes \mathbb{I}_B$. As shown in Fig.\ref{box_9}, the set of effects $\{1,2,3,4,9,10,11,12\}$ remain fixed under this transformation, whereas the pairs $\{5,13\}$,$\{6,14\}$,$\{7,15\}$ and $\{8,16\}$ exchange places among themselves. In order to find the number of fixed points of $r_Af_A\otimes \mathbb{I}_B$ we need to choose a set of $8$ effects that remains invariant under the action of  $r_Af_A\otimes \mathbb{I}_B$. We have the following possibilities:
\begin{figure}[b!]
\centering
\includegraphics[width=0.68\textwidth]{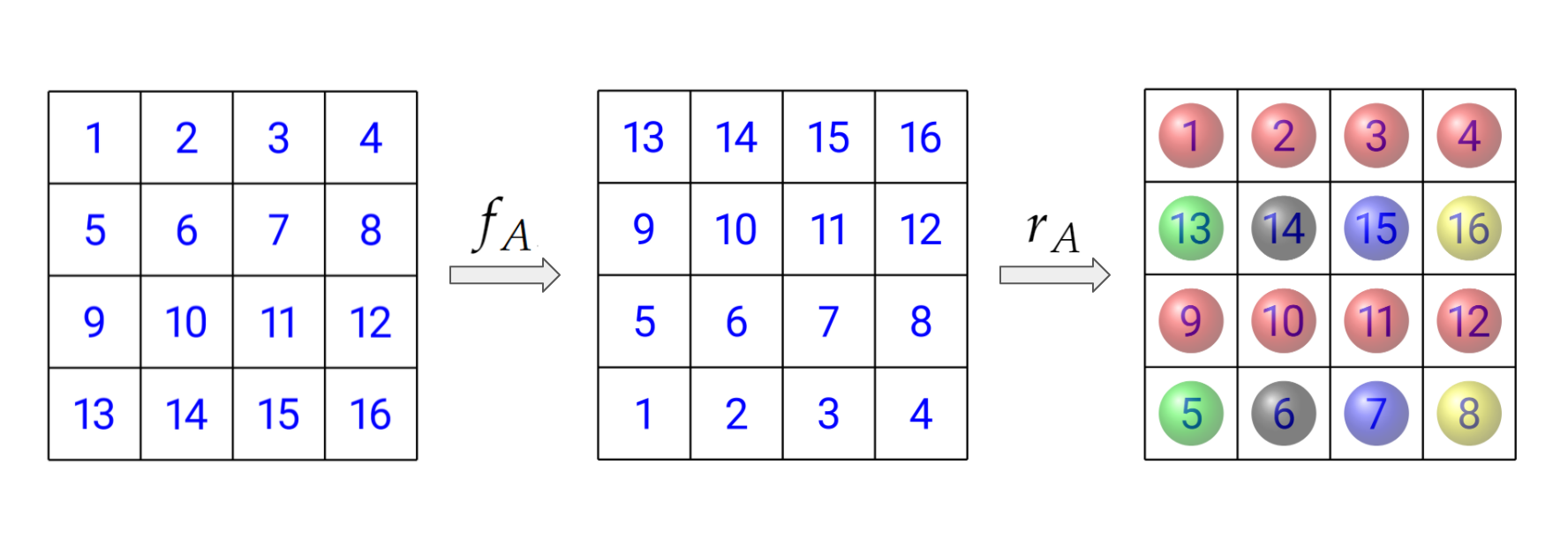}
\caption{(Color online) The effects colored red remain fixed under the action of $r_Af_A\otimes \mathbb{I}_B$. The effects colored green $\{5,13\}$ is a pair of effects which flip to each other upon the action of $r_Af_A\otimes \mathbb{I}_B$. Similarly, the black pair $\{6,14\}$, the blue pair $\{7,15\}$, and the yellow pair $\{8,16\}$ flip to each other under action of $r_Af_A\otimes \mathbb{I}_B$.}\label{box_9}
\end{figure}
\begin{itemize}
\item[1)] choose all the fixed $8$ effects $\implies^8C_8\times^4C_0$ possibilities,
\item[2)] choose $6$ fixed effects and $1$ pair $\implies^8C_6\times^4C_1$ possibilities,
\item[3)] choose $4$ fixed effects and $2$ pairs $\implies^8C_4\times^4C_2$ possibilities,
\item[4)] choose $2$ fixed effects and $3$ pairs $\implies^8C_2\times^4C_3$ possibilities,
\item[5)] choose all $4$ pairs $\implies^8C_0\times^4C_4$ possibilities,
\end{itemize}
Therefore we have a total of $646$ different fixed points for $r_Af_A\otimes \mathbb{I}_B$. For the other group elements in $\mathcal{LR}_{box}$ we can carry a similar procedure to count the number of fixed points, which has been shown in Table \ref{tab4}.

\begin{table}[t!]
\centering
\begin{tabular}{c||c|c|c|c|c|c|c|c|}
%\hline
$\otimes$&~$\mathbb{I}_B$~&~~$r_B$~~&~~$r^2_B$~~&~~$r^3_B$~~ &~~$f_B$~~&~$r_Bf_B$~&~$r^2_Bf_B$~&~$r^3_Bf_B$~\\
\hline\hline
$\mathbb{I}_A$&$12870$&$6$&$70$&$6$&$70$&$646$&$70$&$646$\\
\hline
$r_A$&$6$&$6$&$6$&$6$&$6$&$6$&$6$&$6$\\
\hline
$r^2_A$&$70$&$6$&$70$&$6$&$70$&$70$&$70$&$70$\\
\hline
$r^3_A$&$6$&$6$&$6$&$6$&$6$&$6$&$6$&$6$\\
\hline
$f_A$&$70$&$6$&$70$&$6$&$70$&$70$&$70$&$70$\\
\hline
$r_Af_A$&$646$&$6$&$70$&$6$&$70$&$150$&$70$&$150$\\
\hline
$r^2_Af_A$&$70$&$6$&$70$&$6$&$70$&$70$&$70$&$70$\\
\hline
$r^3_Af_A$&$646$&$6$&$70$&$6$&$70$&$150$&$70$&$150$\\
\hline
\end{tabular}
\caption{Number of fixed point for all the local reversible transformations $g=t_A\otimes t_B\in\mathcal{LR}_{box}.$ }
\label{tab4}
\end{table}
Thus from Table \ref{tab4} we have the total number of orbits which turns out to be 
\begin{align*}
\mbox{\# of orbits}&= \frac{1}{|\mathcal{LR}_{box}|}\sum_{g\in\mathcal{LR}_{box}} \mbox{\# of fixed points of $g$}\\
&= \frac{1}{64}\times\mbox{(sum of all entries in Table \ref{tab4})}\\
&=\frac{18112}{64}=283.
\end{align*}
Following a similar counting procedure for higher gons, we obtain Table \ref{tab2}. 

\bibliographystyle{plain}

\end{document}